\documentclass[12pt]{article}
\usepackage{cite}
\usepackage{amsmath,amssymb,amsfonts,mathtools}
\usepackage{algorithmic}
\usepackage{graphicx}
\usepackage{textcomp}
\usepackage[table]{xcolor}
\usepackage{amsthm}
\usepackage{mathabx}
\usepackage{xurl}
\usepackage{footmisc}
\usepackage{hyperref}
\usepackage{cleveref}
\usepackage{caption}
\usepackage{tablefootnote}
\usepackage{array}
\usepackage{float}
\usepackage{enumerate}
\usepackage{longtable}
\usepackage{verbatim}
\usepackage{wasysym}
\usepackage{indentfirst}
\usepackage{multicol}
\usepackage{stmaryrd}

\usepackage{xpatch}
\makeatletter
\AtBeginDocument{\xpatchcmd{\@thm}{\thm@headpunct{.}}{\thm@headpunct{}}{}{}}
\makeatother

\usepackage{soul}
\usepackage{adjustbox}

\usepackage{aliascnt}

\newtheorem{theorem}{Theorem}[section]

\newaliascnt{lemma}{theorem}
\newtheorem{lemma}[lemma]{Lemma}
\aliascntresetthe{lemma}

\newaliascnt{proposition}{theorem}
\newtheorem{proposition}[proposition]{Proposition}
\aliascntresetthe{proposition}

\newaliascnt{corollary}{theorem}
\newtheorem{corollary}[corollary]{Corollary}
\aliascntresetthe{corollary}

\newaliascnt{conjecture}{theorem}

\aliascntresetthe{conjecture}

\crefname{theorem}{theorem}{theorems}
\Crefname{theorem}{Theorem}{Theorems}

\crefname{lemma}{lemma}{lemmas}
\Crefname{lemma}{Lemma}{Lemmas}

\crefname{proposition}{proposition}{propositions}
\Crefname{proposition}{Proposition}{Propositions}

\crefname{corollary}{corollary}{corollaries}
\Crefname{corollary}{Corollary}{Corollaries}

\crefname{conjecture}{conjecture}{conjectures}
\Crefname{conjecture}{Conjecture}{Conjectures}

\theoremstyle{definition}

\newaliascnt{definition}{theorem}

\aliascntresetthe{definition}

\newaliascnt{example}{theorem}
\newtheorem{example}[example]{Example}
\aliascntresetthe{example}

\newaliascnt{remark}{theorem}
\newtheorem{remark}[remark]{Remark}
\aliascntresetthe{remark}

\crefname{definition}{definition}{definitions}
\Crefname{definition}{Definition}{Definitions}

\crefname{example}{example}{examples}
\Crefname{example}{Example}{Examples}

\crefname{remark}{remark}{remarks}
\Crefname{remark}{Remark}{Remarks}

\usepackage{arydshln}
\usepackage{booktabs}

\usepackage[numbers]{natbib}
\usepackage{colortbl}

\newcommand\ChangeRT[1]{\noalign{\hrule height #1}}  

\DeclareCaptionType{equ}[][]
\theoremstyle{plain}

\begin{document}

\title{Cyclic codes over a commutative non-unitary ring of order $4$\footnotemark[1]}

\author{Jon-Lark Kim\footnotemark[2]   \\ Department of Mathematics and \\ Institute for Mathematical and Data Sciences  \\ Sogang University, Seoul, Korea \\
		{\tt jlkim@sogang.ac.kr } \\
			\\ Marvin Olavides  \\ Department of Mathematics and \\ Institute for Mathematical and Data Sciences  \\ Sogang University, Seoul, Korea \\
		{\tt mmolavides@gmail.com}   
	}
\date{}
	
\maketitle

\footnotetext[1]{J.-L. Kim was supported in part by the BK21 FOUR (Fostering Outstanding Universities for Research) funded by the Ministry of Education (MOE, Korea) and National Research Foundation of Korea (NRF) under Grant No. 4120240415042 and by Basic Science Research Program through the National Research Foundation of Korea (NRF) funded by the Ministry of Science and ICT under Grant No. RS-2025-24534992.}

\footnotetext[2]{Corresponding author}

\begin{abstract}
Let $I_2$ be the commutative non-unitary ring of order $4$ arising in the classification of Fine. In this paper, we investigate cyclic codes over $I_2$ through their associated residue and torsion codes over $\mathbb{F}_2$. We introduce the notions of twisted and untwisted cyclic codes and characterize cyclicity in terms of a compatibility condition involving the twist map and the cyclic shift. Connections between cyclic codes over $I_2$ and binary quasi-cyclic codes are established via Gray maps. In particular, we show that the Gray image of a cyclic code over $I_2$ is a binary quasi-cyclic code of index $2$. We also study duality properties of cyclic codes over $I_2$ and prove that the dual of a cyclic code is again cyclic. Finally, we classify permutation inequivalent cyclic codes over $I_2$ for lengths $n \le 7$ and determine various structural properties of these codes.
\end{abstract}



\section{Introduction}

Cyclic codes are among the most important and extensively studied families of linear codes because of their rich algebraic structure and efficient encoding and decoding algorithms. Their representation as ideals of the polynomial quotient ring $R[x]/(x^n-1)$, where $x$ is an indeterminate and $R$ is a finite field or a finite ring, provides powerful algebraic tools for analyzing their structural properties, constructing codes with good parameters, and developing efficient decoding procedures. Consequently, cyclic codes have played a central role in both the theory and applications of error-correcting codes. Classical treatments of cyclic codes may be found in the monographs of MacWilliams and Sloane \cite{MacWilliamsSloane}, van Lint \cite{vanLint}, and Huffman and Pless \cite{HuffmanPless}.

The discovery of Hammons \textit{et al.} \cite{HammonsKerdPrepCodes} that several remarkable nonlinear binary codes arise as Gray images of linear codes over $\mathbb Z_4$ stimulated extensive research on coding theory over finite rings. Since then, cyclic codes over various classes of finite rings, including finite chain rings, principal ideal rings, Frobenius rings, and local rings, have been widely investigated. In particular, cyclic and negacyclic codes over finite chain rings have been studied extensively because of their rich algebraic structures and their applications to the construction of good binary and quaternary codes through Gray maps; see, for example, \cite{DinhLopezPermouth} and the references therein.

More recently, cyclic codes over finite non-unitary rings have begun to attract attention. In particular, using the rings $E$ and $H$ in the classification of Fine \cite{Fine}, Alahmadi \textit{et al.} investigated cyclic codes over these rings \cite{CyclicE2, CyclicH2}. These works demonstrate that the absence of a multiplicative identity gives rise to algebraic properties that differ substantially from those of cyclic codes over rings with identity. Nevertheless, the theory of cyclic codes over finite non-unitary rings remains far from complete, and many classes of such rings have yet to be investigated.

On the other hand, the present paper focuses on the commutative non-unitary ring $I_2$ which has ring presentation $I_2 =  \left\langle \mathtt{a} , \mathtt{b}   \mid 2 \mathtt{a}  = 2 \mathtt{b} = 0, \mathtt{a}^{2} = \mathtt{b},  \mathtt{a b}=0   \right\rangle$. This ring belongs to the class of finite commutative non-unitary rings classified by Fine \cite{Fine}. Rings of order four have already proved useful in coding theory. For instance, Kim and Ohk \cite{KimJLOhkDE_DNAcodes} constructed DNA codes over two noncommutative rings of order four, illustrating that finite rings of small order can provide rich algebraic frameworks for code constructions. More recently, Kim \textit{et al.} \cite{KimJLOlavMYGRoe_MDSACDI2I3} investigated MDS and ACD codes over commutative non-unitary rings of orders four and nine. Linear codes over $I_2$ were introduced by Alahmadi \textit{et al.} \cite{QuasiTypeIV_I2}, who studied quasi self-dual codes, Type~IV codes, and quasi Type~IV codes over $I_2$ and classified these codes for lengths $n<4$. This classification was subsequently extended to lengths $n\le6$ in \cite{BuildUp_I2} and further to lengths $n\le8$ by Kim \textit{et al.} \cite{KimJLYGRoe_I2}. Despite these developments, cyclic codes over the commutative non-unitary ring $I_2$ have not yet been systematically explored.

Although $I_2$ has the same cardinality as several familiar finite rings, its algebraic structure differs significantly from the non-unitary rings considered in \cite{CyclicH2,CyclicE2}. In particular, every $I_2$-code naturally determines a residue code and a torsion code over $\mathbb F_2$. However, these two binary codes alone do not completely describe the structure of an $I_2$-code. An additional invariant, called the twist map, is required to capture the interaction between the residue and torsion codes. Consequently, residue and torsion codes alone are insufficient to characterize cyclic codes over $I_2$, and the twist map becomes an essential ingredient in their structural theory.

Motivated by these observations, this paper develops a systematic theory of cyclic codes over $I_2$. We introduce the notions of twisted and untwisted cyclic codes through the twist map and establish a characterization theorem showing that an $I_2$-code is cyclic if and only if its residue code and torsion code are binary cyclic codes and the twist map is compatible with the cyclic shift. This characterization provides a structural description of cyclic codes over $I_2$ that does not appear in the existing literature on cyclic codes over finite non-unitary rings. We further establish generator matrix representations for twisted and untwisted cyclic codes and investigate their algebraic properties.

We also investigate Gray images of cyclic codes over $I_2$. In particular, we prove that the Gray image of a cyclic code over $I_2$ is a binary quasi-cyclic code of index two and establish its relationship with additive cyclic codes over $\mathbb F_4$. Furthermore, we derive bounds on the minimum distances of Gray images in terms of the corresponding residue and torsion codes and determine when these bounds are attained. We also prove that the dual of a cyclic code over $I_2$ is again cyclic. Finally, using \texttt{MAGMA} \cite{MAGMA}, we classify all permutation-inequivalent nonzero cyclic codes over $I_2$ of lengths $n\le7$, determine their residue and torsion codes, compute the parameters of their Gray images, and distinguish twisted cyclic codes from untwisted cyclic codes.

The remainder of this paper is organized as follows. \Cref{sec:prelim} reviews the necessary preliminaries on the ring $I_2$, $I_2$-codes, and their associated residue and torsion codes. \Cref{sec:cyclic-I2} develops the theory of cyclic codes over $I_2$, including twisted and untwisted cyclic codes, characterization theorems, generator matrices, and duality results. \Cref{sec: Gray images of I2-codes} investigates Gray images of cyclic codes over $I_2$, establishes their connections with binary quasi-cyclic codes and additive cyclic codes over $\mathbb F_4$, and studies their parameters. \Cref{sec: comp results} presents computational results obtained using \texttt{MAGMA}, including the complete classification of permutation-inequivalent nonzero cyclic codes over $I_2$ for lengths $n\le7$. Finally, \Cref{sec: conclusion} concludes the paper with a summary of the main results and several directions for future research.

\section{Preliminaries}  \label{sec:prelim}

\subsection{Binary codes}

Let $\mathbb{F}_2$ be the finite field of order $2$ and $\mathbb{F}_{2}^{n}$ be the vector space of $n$-tuples over $\mathbb{F}_{2}$. A linear code $\mathcal{C}$ over  $\mathbb{F}_{2}$ is a $k$-dimensional subspace of $\mathbb{F}_{2}^{n}$ and we call $\mathcal{C}$ an $[n, k]$ binary linear code if it has length $n$ and dimension $k$. Throughout, when we say binary codes, we mean binary linear codes. Given $\mathbf{x} = (x_1 , \dotsc , x_n)$ and $\mathbf{y} = ( y_1 , \dotsc , y_n) \in \mathbb{F}_2^n$, the standard inner product is denoted by $ \mathbf{x} \cdot \mathbf{y}  \coloneqq \sum_{i=1}^{n}  x_i y _i$. The dual of a  binary code is denoted by $\mathcal{C}^{\perp}$ and is defined as
\begin{align*}
\mathcal{C}^{\perp} \coloneqq \left\{  \mathbf{y} \in \mathbb{F}_{2}^{n} \mid \forall \mathbf{x} \in \mathcal{C},  \  \mathbf{x}  \cdot \mathbf{y}  =  0 \right\}  .
\end{align*}
A code $\mathcal{C}$ is self-orthogonal if  $\mathcal{C} \subseteq \mathcal{C}^{\perp}$ and is self-dual if $\mathcal{C} = \mathcal{C}^{\perp}$.

The (Hamming) weight of a vector $\mathbf{x} \in \mathbb{F}_2^n$ is the number of its nonzero coordinates and is denoted by $\operatorname{wt}_H (\mathbf{x})$. A code is said to be even if all of its codewords have even weights. For $\mathbf{x}, \mathbf{y} \in \mathbb{F}_{2}^{n}$, their  Hamming distance is defined by $d(\mathbf{x}, \mathbf{y}) = \operatorname{wt}_H (\mathbf{x}-\mathbf{y})$. The (minimum) distance $d(\mathcal{C})$ of a binary code $\mathcal{C}$ is
\begin{align*}
d(\mathcal{C}) = \min \{d(\mathbf{x}, \mathbf{y}) \mid \mathbf{x}, \mathbf{y} \in \mathcal{C}, \mathbf{x} \neq \mathbf{y}\}=\min \{ \operatorname{wt}_H  (\mathbf{c}) \mid \mathbf{c} \in \mathcal{C}, \mathbf{c} \neq \mathbf{0}\}  .
\end{align*}

A  linear code over $\mathbb{F}_2$ is called an $[n, k, d]$ binary code if it has length $n$, dimension $k$, and distance $d$. Two binary codes $\mathcal{C}$ and $\mathcal{C}'$ are said to be permutation-equivalent if there is an $n \times n$ permutation matrix $P$ such that $\mathcal{C}' = \{ \mathbf{c} P \, : \, \mathbf{c} \in \mathcal{C}    \}$. An $[n, k, d]$ binary code is said to be optimal if $d = \max\{ d' \mathrel{:} \text{there exists a linear $[n, k , d']$ binary code}   \}$.

A binary code $\mathcal{C}$ of length $n$ is called a cyclic code over $\mathbb{F}_2$ given that for each $( c_0 , c_1 , \dotsc , c_{n-2} , c_{n-1} )$ in $\mathcal{C}$, $( c_{n-1} , c_0 , \dotsc , c_{n-2}  )$ is also in $\mathcal{C}$.

\subsection{$\boldsymbol{I_2}$-codes}  \label{subsec: I2}

In this section, we present codes over a non-unital ring of order $4$ which we denote by $I_2$. This notation is consistent with the notation in \cite{Fine}.

Let $I_2$ be the ring with two generators $\mathtt{a}$ and $\mathtt{b}$ and ring presentation
\begin{align*}
I_{2} \coloneqq \left\langle \mathtt{a} , \mathtt{b}   \mid 2 \mathtt{a}  = 2 \mathtt{b} = 0, \mathtt{a}^{2} = \mathtt{b},  \mathtt{a b}=0   \right\rangle  .
\end{align*}
Thus,  $I_{2}$ has  characteristic $2$ and consists of $4$ elements, i.e., $I_2 = \{ 0 , \mathtt{a} , \mathtt{b} , \mathtt{c}  \}$ where $\mathtt{c} = \mathtt{a} + \mathtt{b}$. The addition and the multiplication tables for $I_2$ are given below.


\begin{table}[H]
\caption{Addition table for the ring $I_2$}  \label{tbl:add table I2}
\centering
\scalebox{0.9}{
\begin{tabular}{ c ! {\vrule width 1.5pt} c|c|c|c}

$+$ & $0$ & $\mathtt{a}$ & $\mathtt{b}$ & $\mathtt{c}$  \\
\ChangeRT{1.5pt}
$0$ & $0$ & $\mathtt{a}$ & $\mathtt{b}$ & $\mathtt{c}$  \\
\hline
$\mathtt{a}$ & $\mathtt{a}$  &  $0$ &  $\mathtt{c}$  &  $\mathtt{b}$   \\
\hline
$\mathtt{b}$ & $\mathtt{b}$  &  $\mathtt{c}$ &  $0$  &  $\mathtt{a}$  \\
\hline
$\mathtt{c}$ & $\mathtt{c}$  &  $\mathtt{b}$ &  $\mathtt{a}$  &  $0$
\end{tabular}
}

\end{table}


\begin{table}[H]
\caption{Multiplication table for the ring $I_2$}  \label{tbl:mult table I2}
\centering
\scalebox{0.9}{
\begin{tabular}{ c ! {\vrule width 1.5pt} c|c|c|c}

$\times$ & $0$ & $\mathtt{a}$ & $\mathtt{b}$ & $\mathtt{c}$  \\
\ChangeRT{1.5pt}
$0$ & $0$ & $0$ & $0$ & $0$  \\
\hline
$\mathtt{a}$ & $0$  &  $\mathtt{b}$ &  $0$  &  $\mathtt{b}$   \\
\hline
$\mathtt{b}$ & $0$  &  $0$ &  $0$  &  $0$  \\
\hline
$\mathtt{c}$ & $0$  &  $\mathtt{b}$ &  $0$  &  $\mathtt{b}$
\end{tabular}
}

\end{table}


We can infer from Table \ref{tbl:mult table I2} that this ring is commutative without unity and has a unique maximal ideal $J = \{ j  \mathtt{b}: 0 \leq j < 2 \}  = \{0 ,  \mathtt{b}   \}$ with residue field $I_{2} / J \cong \mathbb{F}_{2}$. It is worth noting that each element $r \in I_2$ can be expressed as $r = \mathtt{a} x + \mathtt{b} y$ where $x , y \in \mathbb{F}_2$.

We define the reduction map modulo $J$ as $\alpha \colon I_{2} \longrightarrow I_{2} / J \simeq \mathbb{F}_{2}$ by $\alpha (0) = \alpha (\mathtt{b}) = 0$ and $\alpha (\mathtt{a}) = \alpha (\mathtt{c}) = 1$. The map $\alpha$ is extended in the natural way to a map from $I_{2}^{n}$ to $\mathbb{F}_{2}^{n}$.

A linear $I_2$-code of length $n$ is defined as an $I_{2}$-submodule of $I_{2}^{n}$. Here, we will also drop the word ``linear'' and we just speak of linear $I_2$-codes as simply $I_2$-codes. An additive code $\mathcal{C}$ of length $n$ over $\mathbb{F}_{4}$ is defined as an additive subgroup of $\mathbb{F}_{4}^{n}$. To an $I_2$-code $\mathcal{C}$, we attach an additive $\mathbb{F}_4$-code $\phi (\mathcal{C})$ via the map $\phi \colon I_2 \to \mathbb{F}_4$ where
\begin{gather*}
\phi (0) = 0 ,   \quad     \phi ( \mathtt{a} ) =  \omega ,     \quad    \phi (   \mathtt{b}  ) =    1 , \quad   \phi (  \mathtt{c}  )  =   \omega^2
\end{gather*}
and $\mathbb{F}_4 = \mathbb{F}_2 [\omega]$, extended naturally to $\mathbb{F}_4^n$.

Two binary codes of length $n$ can be associated canonically with an $I_2$-code $\mathcal{C}$. The residue code $\operatorname{res}(\mathcal{C})$ is defined as
\begin{align*}
\operatorname{res} (\mathcal{C})  \coloneqq  \left\lbrace   \alpha ( \mathbf{y} )  \mid  \mathbf{y}  \in \mathcal{C}  \right \rbrace  =   \alpha (\mathcal{C} )
\end{align*}
and the torsion code  $\operatorname{tor}(\mathcal{C})$  is defined as
\begin{align*}
\operatorname{tor}(\mathcal{C}) \coloneqq \left\{\mathbf{x} \in \mathbb{F}_{2}^{n}  \mid   \mathtt{b}  \mathbf{x} \in \mathcal{C}  \right\}  .
\end{align*}
Note that $\operatorname{res}(\mathcal{C}) \subseteq \operatorname{tor}(\mathcal{C})$. An $I_2$-code $\mathcal{C}$ is said to be of type $\left \{  k_{1}, k_{2}  \right  \}$ if $\operatorname{res} (\mathcal{C})$ has dimension $k_1$ and $\operatorname{tor} (\mathcal{C})$ has dimension $k_1 + k_2$. We say that an $I_2$-code $\mathcal{C}$ is free if and only if $k_2 = 0$. Let $\alpha_\mathcal{C}$ be the restriction of $\alpha$ to $\mathcal{C}$. Given an $I_2$-code $\mathcal{C}$ of type $\{ k_1 , k_2\}$, the first isomorphism theorem applied to $\alpha_\mathcal{C}$ gives
\begin{align*}
|\mathcal{C}|=|\operatorname{res}(\mathcal{C})||\operatorname{tor}(\mathcal{C})|=2^{2 k_{1}+k_{2}} .
\end{align*}
Given $\mathbf{x} = (x_1 , \dotsc , x_n)$ and $\mathbf{y} = ( y_1 , \dotsc , y_n) \in I_2^n$, the inner product on $I_2^n$ is defined to be $ \mathbf{x} \cdot \mathbf{y}  \coloneqq \sum_{i=1}^{n} x_i y_i$. The dual code $\mathcal{C}^{\perp}$ of the $I_2$-code $\mathcal{C}$ is the module defined as
\begin{align*}
\mathcal{C}^{\perp}=\left\{\mathbf{y} \in I_{2}^{n} \mid \forall \mathbf{x} \in \mathcal{C} ,  \  \mathbf{x}  \cdot \mathbf{y} = 0   \right\}  .
\end{align*}

An $I_{2}$-code $\mathcal{C}$ is self-orthogonal (SO) if for all $\mathbf{x}, \mathbf{y} \in \mathcal{C}$, $\mathbf{x}  \cdot \mathbf{y} = 0$. If $\mathcal{C}=\mathcal{C}^{\perp}$, we say that $\mathcal{C}$ is self-dual (SD). If $\mathcal{C}$ is self-orthogonal and has size $2^n$, we say that $\mathcal{C}$ is quasi self-dual (QSD). A quasi self-dual code $\mathcal{C}$ over $I_2$ where each  codeword has even weight is called a type IV code. A quasi self-dual $I_2$-code with an even torsion code is called a quasi type IV (QTIV) code. Two $I_{2}$-codes $\mathcal{C}$ and $\mathcal{C}'$ are said to be permutation-equivalent if there is an $n \times n$ permutation matrix $P$ such that $\mathcal{C}' = \{ \mathbf{c} P \mathrel{:} \mathbf{c} \in \mathcal{C}    \}$.

\section{Cyclic Codes over $\boldsymbol{I_2}$}    \label{sec:cyclic-I2}

We adopt the definition of cyclic codes over finite fields and apply it to the ring $I_2$.

A cyclic code $\mathcal{C}$ of length $n$ over $I_2$ is an $I_2$-code with the property that if $\mathbf{c} = ( c_0 , c_1 , \dotsc , c_{n-2} , c_{n-1} )$ is in $\mathcal{C}$, then $\tilde{\mathbf{c}} = ( c_{n-1} , c_0 , \dotsc , c_{n-2}  )$ is also in $\mathcal{C}$.

We will be flexible in using the cyclic shift $\sigma$. That is, if  $\mathbf{x} = ( x_0 , x_1 , \dotsc , x_{n-2} , x_{n-1} )$ is an $n$-tuple vector, then $\sigma ( \mathbf{x} ) = ( x_{n-1} , x_0 , \dotsc , x_{n-2}  )$. With this, an $I_2$-code $\mathcal{C}$ is cyclic if $\sigma (\mathcal{C}) = \mathcal{C}$.

Let $\mathcal{C} \subseteq I_2^n$ be an $I_2$-code. Define the twist map
\begin{align*}
\tau \colon  \operatorname{res} (\mathcal{C})  \to  \mathbb{F}_2^n /  \operatorname{tor}  ( \mathcal{C}  ) ,  \qquad      \tau  ( \mathbf{r} ) =   \mathbf{s}_\mathbf{r}  +  \operatorname{tor}  ( \mathcal{C} )
\end{align*}
where $\mathbf{s}_\mathbf{r}  \in  \mathbb{F} _2^n$ is any vector for which $ \mathtt{a}  \mathbf{r}  +  \mathtt{b}  \mathbf{s}_\mathbf{r}  \in \mathcal{C}$. Note that the coset $ \tau ( \mathbf{r} )$ is well-defined and is independent of the choice of $ \mathbf{s}_\mathbf{r} $. Then
\begin{align*}
\mathcal{C}
=
\{  \mathtt{a} \mathbf{r}  +  \mathtt{b}  ( \mathbf{s}_\mathbf{r}  + \mathbf{t}   ) \mathrel{:}  \mathbf{r}  \in  \operatorname{res}  (\mathcal{C}) ,  \,  \mathbf{t}  \in  \operatorname{tor} (\mathcal{C})   \}  .
\end{align*}
It can be checked that the map $\tau$ is $\mathbb{F}_2$-linear and the set of $I_2$-codes is in a one-to-one correspondence with the set of triples $( \operatorname{res}  (\mathcal{C}) , \operatorname{tor}  (\mathcal{C}),  \tau )$.

\begin{lemma}  \label{lem: untwisted I2-codes equiv statements}
Let $\mathcal{C}$ be an $I_2$-code of length $n$. The following are equivalent:
\begin{enumerate}[(i)]
\item  $\tau ( \mathbf{r} ) =  \operatorname{tor}  (\mathcal{C})$ for all $ \mathbf{r}  \in \operatorname{res} (\mathcal{C})$  \label{itm: untwisted 1}
\item  $\mathbf{s}_\mathbf{r}  \in  \operatorname{tor}  (\mathcal{C})$ for all $ \mathbf{r}  \in \operatorname{res} (\mathcal{C})$  \label{itm: untwisted 2}
\item  $\mathcal{C} =  \mathtt{a} \operatorname{res} (\mathcal{C})  +  \mathtt{b}  \operatorname{tor}  (\mathcal{C})$  \label{itm: untwisted 3}
\end{enumerate}
\end{lemma}

\begin{proof}
(\ref{itm: untwisted 1}) $\Rightarrow$ (\ref{itm: untwisted 2}): Assume that $\tau (\mathbf{r}) = \operatorname{tor}(\mathcal{C})$ for all $\mathbf{r} \in \operatorname{res}(\mathcal{C})$. By the definition of $\tau$, we must have $\mathbf{s}_{\mathbf{r}} + \operatorname{tor}(\mathcal{C}) = \operatorname{tor}(\mathcal{C})$, and this holds when $\mathbf{s}_{\mathbf{r}} \in \operatorname{tor}(\mathcal{C})$. Thus, $\mathbf{s}_\mathbf{r}  \in  \operatorname{tor}  (\mathcal{C})$ for all $ \mathbf{r}  \in \operatorname{res} (\mathcal{C})$.

(\ref{itm: untwisted 2}) $\Rightarrow$ (\ref{itm: untwisted 3}):  Since $\mathbf{s}_\mathbf{r}  \in  \operatorname{tor}  (\mathcal{C})$ for all $ \mathbf{r}  \in \operatorname{res} (\mathcal{C})$, then $\mathbf{t}' \coloneqq  \mathbf{s}_{\mathbf{r}} + \mathbf{t} \in \operatorname{tor} ( \mathcal{C} )$. Thus, $
\mathcal{C}
=
\{  \mathtt{a} \mathbf{r}  +  \mathtt{b}  ( \mathbf{s}_\mathbf{r}  + \mathbf{t}   )  \mathrel{:}  \mathbf{r}  \in  \operatorname{res}  (\mathcal{C}) ,  \,  \mathbf{t}  \in  \operatorname{tor} (\mathcal{C})   \}
=
\{  \mathtt{a} \mathbf{r}  +  \mathtt{b}  \mathbf{t}'   \mathrel{:}  \mathbf{r}  \in  \operatorname{res}  (\mathcal{C}) ,  \,  \mathbf{t}'  \in  \operatorname{tor} (\mathcal{C})   \}
= \mathtt{a} \operatorname{res} (\mathcal{C})  +  \mathtt{b}  \operatorname{tor}  (\mathcal{C})
$.

(\ref{itm: untwisted 3}) $\Rightarrow$ (\ref{itm: untwisted 1}):  Assume that $\mathcal{C} = \mathtt{a}\operatorname{res}(\mathcal{C}) + \mathtt{b}\operatorname{tor}(\mathcal{C})$. Let $\mathbf{r} \in \operatorname{res}(\mathcal{C})$. Then $\mathtt{a}\mathbf{r} = \mathtt{a}\mathbf{r} + \mathtt{b}\mathbf{0} \in \mathcal{C}$. Thus, $\mathbf{0}$ is a valid choice for $\mathbf{s}_{\mathbf{r}}$ in the definition of $\tau (\mathbf{r})$, and so $\tau (\mathbf{r}) = \mathbf{0} + \operatorname{tor}(\mathcal{C}) = \operatorname{tor}(\mathcal{C})$. Since $\mathbf{r} \in \operatorname{res}(\mathcal{C})$ is arbitrary, it follows that $\tau ( \mathbf{r} ) =  \operatorname{tor}  (\mathcal{C})$ for all $ \mathbf{r}  \in \operatorname{res} (\mathcal{C})$.
\end{proof}

We say that an $I_2$-code $\mathcal{C}$ is \textit{untwisted} if $\mathcal{C}$ satisfies any of the equivalent conditions in \Cref{lem: untwisted I2-codes equiv statements}; otherwise, we say that $\mathcal{C}$ is \textit{twisted}. We can give a stronger version of statement (\ref{itm: untwisted 3}) of \Cref{lem: untwisted I2-codes equiv statements} as seen in the lemma below.

\begin{lemma}  \label{thm:C as direct sum of resC and torC}
If $\mathcal{C}$ is an untwisted $I_2$-code, then $\mathcal{C} = \mathtt{a} \operatorname{res} (\mathcal{C}) \oplus \mathtt{b} \operatorname{tor} (\mathcal{C})$.
\end{lemma}

\begin{proof}
From statement (\ref{itm: untwisted 3}) of \Cref{lem: untwisted I2-codes equiv statements}, $\mathcal{C} = \mathtt{a} \operatorname{res} (\mathcal{C}) + \mathtt{b} \operatorname{tor} (\mathcal{C})$. Now, if $\mathtt{a} \mathbf{r} = \mathtt{b} \mathbf{t}$ where $ \mathbf{r} \in \operatorname{res} (\mathcal{C})$ and $ \mathbf{t} \in\operatorname{tor}( \mathcal{C} )$, then $\alpha ( \mathtt{a} \mathbf{r} )  = \alpha ( \mathtt{b}  \mathbf{t} ) = 0$ since $\alpha (  \mathtt{a} ) = 1$ and $\alpha ( \mathtt{b} ) = 0$. Hence, $\mathbf{r} = \mathbf{0}$, and so $ \mathtt{b} \mathbf{t} = \mathbf{0}$, implying that $\mathbf{t} = \mathbf{0}$. Thus, $ \mathtt{a} \operatorname{res}(\mathcal{C}) \cap \mathtt{b} \operatorname{tor}(\mathcal{C}) = \{ \mathbf{0} \}$. So, the sum is direct, i.e., $\mathcal{C} = \mathtt{a} \operatorname{res} (\mathcal{C}) \oplus \mathtt{b} \operatorname{tor} (\mathcal{C})$.
\end{proof}

\begin{theorem}  \label{thm:gen-matrix-cyclic-I2-codes}
Let $\mathcal{C}$ be a cyclic code over $I_2$ of length $n$ and type $\{ k_1 , k_2 \}$. Let $G_1$ and $G_2$ be binary generator matrices of $\operatorname{res}( \mathcal{C} )$ and $\operatorname{tor}( \mathcal{C} )$, respectively. For each row $\mathbf{r}$ of $G_1$, choose a vector $\mathbf{s}_{\mathbf{r}} \in \mathbb{F}_2^n$ such that
$
\mathtt{a}\mathbf{r} + \mathtt{b}\mathbf{s}_{\mathbf{r}} \in \mathcal{C}
$.
Let $S$ be the matrix whose rows are the vectors $\mathbf{s}_{\mathbf{r}}$. Then a generator matrix of $\mathcal{C}$ is
\begin{align*}
G = \begin{pmatrix}
\mathtt{a}G_1 + \mathtt{b} S  \\ 
\mathtt{b}G_2
\end{pmatrix}
\in I_2^{ ( 2k_1 + k_2 ) \times n}  .
\end{align*}
Moreover, if $\mathcal{C}$ is untwisted, $G$ can be reduced to
\begin{align*}
G = \begin{pmatrix} \mathtt{a}G_1 \\
 \mathtt{b}G_2 \end{pmatrix}.
\end{align*}
\end{theorem}

\begin{proof}
From the structural decomposition of $I_2$-codes, every codeword $\mathbf{c}\in\mathcal{C}$ can be written uniquely as
$
\mathbf{c} = \mathtt{a} \mathbf{r} + \mathtt{b} ( \mathbf{s}_{\mathbf{r}} + \mathbf{t} )
$
where $\mathbf{r} \in \operatorname{res}(\mathcal{C})$, $\mathbf{t} \in \operatorname{tor} (\mathcal{C})$, and $\mathbf{s}_{\mathbf{r}}$ is a representative of the coset $\tau (\mathbf{r})$.
Since $G_1$ and $G_2$ generate $\operatorname{res}(\mathcal{C})$ and $\operatorname{tor}(\mathcal{C})$ as binary codes, the rows of $\mathtt{a} G_1 + \mathtt{b}S$ generate all vectors of the form $\mathtt{a} \mathbf{r} + \mathtt{b} \mathbf{s}_{\mathbf{r}}$ and the rows of $\mathtt{b}G_2$ generate all vectors of the form $\mathtt{b} \mathbf{t}$.
Thus, the $I_2$-span of the rows of $G$ contains every element of $\mathcal{C}$ and therefore equals $\mathcal{C}$. Now, if $\mathcal{C}$ is untwisted, then we can take $\mathbf{s}_{\mathbf{r}} = \mathbf{0}$ for all $\mathbf{r} \in \operatorname{res} ( \mathcal{C} )$. Thus, $S = \mathbf{0}$, and we get $G = \begin{pmatrix} \mathtt{a}G_1 \\
 \mathtt{b}G_2 \end{pmatrix}$.
\end{proof}

\begin{remark}
The matrix $S$ captures the deviation from the direct sum decomposition $\mathcal{C}=\mathtt{a}\operatorname{res}(\mathcal{C}) \oplus \mathtt{b}\operatorname{tor}(\mathcal{C})$ and is therefore referred to as the \emph{twist matrix} of $\mathcal{C}$.
When $S\neq \mathbf{0}$, the residue and torsion components no longer separate cleanly, even though the cyclic shift still preserves the module structure. The case $S = \mathbf{0}$ corresponds precisely to the untwisted codes.
\end{remark}

\begin{example}  \label{ex: cyclic code of length 4}
Let $\mathcal{C}$ be the cyclic code of length $n = 4$ and type $\{ 1  , 1 \}$ given by $\mathcal{C} = \{ (0,0,0,0), ( \mathtt{b} , 0 , \mathtt{b} , 0  ) , ( 0, \mathtt{b} , 0 , \mathtt{b} ), ( \mathtt{b} , \mathtt{b} , \mathtt{b} , \mathtt{b})  , ( \mathtt{c} , \mathtt{c} , \mathtt{a} , \mathtt{a} ) , ( \mathtt{a} , \mathtt{c} , \mathtt{c} , \mathtt{a} ) , ( \mathtt{c} , \mathtt{a} , \mathtt{a} , \mathtt{c} ) , ( \mathtt{a} , \mathtt{a} , \mathtt{c} , \mathtt{c} )  \}$ over $I_2$. Then $\operatorname{res} ( \mathcal{C} ) = \alpha ( \mathcal{C} ) = \{ (0,0,0,0), (1,1,1,1) \}$ and $\operatorname{tor} ( \mathcal{C} ) = \{ (0,0,0,0) , (1,0,1,0),$ $(0,1,0,1) , (1,1,1,1) \}$. Thus,
\begin{align*}
G_1 = \left(
\begin{array}{cccc}
1 & 1 & 1 & 1
\end{array}
\right)
\end{align*}
and
\begin{align*}
G_2 = \left(
\begin{array}{cccc}
1 & 0 & 1 & 0  \\
0 & 1 & 0 & 1
\end{array}
\right)
\end{align*}
are generator matrices for $\operatorname{res} ( \mathcal{C} )$ and $\operatorname{tor} ( \mathcal{C} )$, respectively. Observe that $G_1$ is just a row vector, so if we take $\mathbf{s_r} = (1,1,0,0) \in \mathbb{F}_2^4$, we see that $\mathtt{a} \mathbf{r} + \mathtt{b} \mathbf{s_r} = ( \mathtt{c} , \mathtt{c} , \mathtt{a} , \mathtt{a} ) \in \mathcal{C}$ where $\mathbf{r} = (1,1,1,1)$. Thus, a twist matrix $S$ for the code $\mathcal{C}$ is given by $S = \left(
\begin{array}{cccc}
1 & 1 & 0 & 0
\end{array}
\right)$, i.e., $\mathcal{C}$ is twisted. By \Cref{thm:gen-matrix-cyclic-I2-codes}, a generator matrix of $\mathcal{C}$ is given by
\begin{align*}
G = \left(
\begin{array}{cccc}
\mathtt{c} & \mathtt{c} & \mathtt{a} & \mathtt{a}  \\
\mathtt{b} & 0 & \mathtt{b} & 0  \\
0 & \mathtt{b} & 0 & \mathtt{b}
\end{array}
\right) .
\end{align*}
\end{example}

We now give a characterization for  cyclic codes over $I_2$.

\begin{theorem}  \label{thm:cyclic code characterization}
An $I_2$-code $\mathcal{C}$ is a cyclic code if and only if both $\operatorname{res} (\mathcal{C})$ and $\operatorname{tor} (\mathcal{C})$ are binary cyclic  codes and $\tau (\sigma (\mathbf{r})) = \sigma ( \tau (\mathbf{r})  )$ for all $\mathbf{r} \in \operatorname{res} (\mathcal{C})$.
\end{theorem}

\begin{proof}
For the forward implication, assume that $\mathcal{C}$ is a cyclic code over $I_2$.
Let $\mathbf{r} \in \operatorname{res} (\mathcal{C})$. Then, by definition of $\operatorname{res}(\mathcal{C})$, there exists a codeword $\mathbf{c} \in \mathcal{C}$ such that $\alpha(\mathbf{c}) = \mathbf{r}$. Thus, $\mathbf{c} = \mathtt{a} \mathbf{r} + \mathtt{b} ( \mathbf{s}_{\mathbf{r}}  + \mathbf{t} )$ where $\mathbf{t} \in \operatorname{tor} ( \mathcal{C} )$ and $\mathbf{s}_{\mathbf{r}}$ is a chosen coset representative. Since $\mathcal{C}$ is cyclic,
\begin{align*}
\sigma(\mathbf{c})
= \sigma ( \mathtt{a} \mathbf{r} + \mathtt{b} ( \mathbf{s}_{\mathbf{r}}  + \mathbf{t}   )  )
= \mathtt{a} \sigma(\mathbf{r}) + \mathtt{b} \sigma( (  \mathbf{s}_{\mathbf{r}}  + \mathbf{t}   )   ) \in \mathcal{C},
\end{align*}
and therefore,
$
\alpha(\sigma(\mathbf{c})) = \sigma(\mathbf{r}) \in \operatorname{res}(\mathcal{C})
$.
Thus, $\operatorname{res}(\mathcal{C})$ is a binary cyclic code.

Now, let $\mathbf{t} \in \operatorname{tor}(\mathcal{C})$. Thus, $\mathtt{b} \mathbf{t} \in \mathcal{C}$. Since $\mathcal{C}$ is cyclic,
$
\sigma( \mathtt{b} \mathbf{t}) = \mathtt{b} \sigma ( \mathbf{t} ) \in \mathcal{C}
$,
and hence $\sigma(\mathbf{t}) \in \operatorname{tor}(\mathcal{C})$. Therefore, $\operatorname{tor}(\mathcal{C})$ is also a binary cyclic code.

Now, fix  $\mathbf{r} \in \operatorname{res} (\mathcal{C})$ and consider
\begin{align*}
K_{\mathbf{r}}
:=
( \mathtt{a} \mathbf{r} + \mathtt{b} \mathbf{s}_{\mathbf{r}}   )
+ \mathtt{b} \operatorname{tor} (\mathcal{C})
=
\{
\mathtt{a} \mathbf{r} + \mathtt{b} (\mathbf{s}_{\mathbf{r}} + \mathbf{t})
\mathrel{:}  \mathbf{t} \in \operatorname{tor}(\mathcal{C})
\}  .
\end{align*}
Applying $\sigma$ gives
\begin{align*}
\sigma(K_{\mathbf{r}})
&= \{
\sigma( \mathtt{a} \mathbf{r} + \mathtt{b} (\mathbf{s}_{\mathbf{r}} + \mathbf{t} ) )
\mathrel{:}  \mathbf{t} \in \operatorname{tor}(\mathcal{C})
\}  \\
&=
\{
\mathtt{a}\sigma(\mathbf{r}) + \mathtt{b} \sigma(\mathbf{s}_{\mathbf{r}} + \mathbf{t})
\mathrel{:}  \mathbf{t} \in \operatorname{tor}(\mathcal{C})
\} \\
&=
\{
\mathtt{a}\sigma(\mathbf{r}) + \mathtt{b} (\sigma(\mathbf{s}_{\mathbf{r}}) + \sigma(\mathbf{t}) )
\mathrel{:}  \mathbf{t} \in \operatorname{tor}(\mathcal{C})
\} .
\end{align*}
Since $\operatorname{tor}(\mathcal{C})$ is cyclic, $\sigma(\mathbf{t})$ runs through
$\operatorname{tor}(\mathcal{C})$, and so
\begin{align*}
\sigma(K_{\mathbf{r}})
=
\{
\mathtt{a}\sigma(\mathbf{r}) + \mathtt{b}\big(\sigma(\mathbf{s}_{\mathbf{r}}) + \mathbf{t}'\big)
\mathrel{:}   \mathbf{t}' \in \operatorname{tor}(\mathcal{C})
\} .
\end{align*}
On the other hand,
\begin{align*}
K_{\sigma(\mathbf{r})}
=
\{
\mathtt{a} \sigma(\mathbf{r}) + \mathtt{b} (\mathbf{s}_{\sigma(\mathbf{r})} + \mathbf{t}'' )
\mathrel{:}   \mathbf{t}'' \in \operatorname{tor}(\mathcal{C})
\} .
\end{align*}
Therefore, $\sigma(K_{\mathbf{r}}) = K_{\sigma(\mathbf{r})}$ if and only if
$
\sigma(\mathbf{s}_{\mathbf{r}})
\equiv
\mathbf{s}_{\sigma(\mathbf{r})}
\pmod{\operatorname{tor}(\mathcal{C})}
$
which is equivalent to
\begin{align*}
\tau (\sigma(\mathbf{r}))
=
\mathbf{s}_{\sigma(\mathbf{r})} + \operatorname{tor}(\mathcal{C})
=
\sigma(\mathbf{s}_{\mathbf{r}}) + \operatorname{tor}(\mathcal{C})
=
\sigma( \tau (\mathbf{r})).
\end{align*}

For the backward implication, assume that $\operatorname{res}(\mathcal{C})$ and $\operatorname{tor}(\mathcal{C})$ are cyclic and that
$\tau (\sigma(\mathbf{r})) = \sigma( \tau (\mathbf{r}))
$ for all $\mathbf{r} \in \operatorname{res}(\mathcal{C})$.
Let $\mathbf{c} \in \mathcal{C}$. Then
$
\mathbf{c}
=
\mathtt{a}\mathbf{r} + \mathtt{b}(\mathbf{s}_{\mathbf{r}} + \mathbf{t})
$
where $\mathbf{r} \in \operatorname{res}(\mathcal{C})$ and $\mathbf{t} \in \operatorname{tor}(\mathcal{C})$.
Hence,
\begin{align*}
\sigma(\mathbf{c})
=
\mathtt{a}\sigma(\mathbf{r}) + \mathtt{b}\sigma(\mathbf{s}_{\mathbf{r}} + \mathbf{t})
=
\mathtt{a}\sigma(\mathbf{r}) + \mathtt{b}\sigma(\mathbf{s}_{\mathbf{r}}) + \mathtt{b}\sigma(\mathbf{t}).
\end{align*}
Since $\sigma(\mathbf{t}) \in \operatorname{tor}(\mathcal{C})$ and
$
\sigma(\mathbf{s}_{\mathbf{r}})
=
\mathbf{s}_{\sigma(\mathbf{r})} + \mathbf{t}_0
$
for some $\mathbf{t}_0 \in \operatorname{tor}(\mathcal{C})$,
we obtain
\begin{align*}
\sigma(\mathbf{c})
=
\mathtt{a}\sigma(\mathbf{r})
+ \mathtt{b}\big( \mathbf{s}_{\sigma(\mathbf{r})} + \mathbf{t}_0 + \sigma(\mathbf{t}) \big).
\end{align*}
Setting
$
\mathbf{t}'
\coloneqq
\mathbf{t}_0 + \sigma(\mathbf{t})
\in \operatorname{tor}(\mathcal{C})
$,
we conclude that
\begin{align*}
\sigma(\mathbf{c})
=
\mathtt{a}\sigma(\mathbf{r}) + \mathtt{b}\big( \mathbf{s}_{\sigma(\mathbf{r})} + \mathbf{t}' \big)
\in \mathcal{C}.
\end{align*}
Thus $\sigma(\mathcal{C}) \subseteq \mathcal{C}$. Since $\sigma$ is bijective,
$\mathcal{C}$ is cyclic.
\end{proof}

\begin{example}
Consider the $I_2$-code $\mathcal{C}$ with generator matrix
\begin{align*}
G=
\begin{pmatrix*}[c]
\mathtt{c} & 0 & \mathtt{a} & 0 \\
\mathtt{b} & \mathtt{a} & 0 & \mathtt{a} \\
\mathtt{b} & 0 & 0 & \mathtt{b} \\
0 & \mathtt{b} & 0 & \mathtt{b} \\
0 & 0 & \mathtt{b} & \mathtt{b}
\end{pmatrix*}.
\end{align*}
We have
\begin{align*}
\operatorname{res}(\mathcal{C})
&=
\left\langle
1010,\,
0101
\right\rangle
=
\{
0000,\,
1010,\,
0101,\,
1111
\}
\end{align*}
and
\begin{align*}
\operatorname{tor}(\mathcal{C})
&=
\left\langle
1001,\,
0101,\,
0011
\right\rangle
=
\{
0000,\,
1001,\,
0101,\,
0011,\,
1100,\,
1010,\,
0110,\,
1111
\}.
\end{align*}
Thus, both $\operatorname{res}(\mathcal{C})$ and $\operatorname{tor}(\mathcal{C})$ are binary cyclic codes. Now, let
$
\mathbf{r}_1=1010
$
and
$
\mathbf{r}_2=0101
$.
Choose
$
\mathbf{s}_{\mathbf{r}_1}=1000
$
and
$
\mathbf{s}_{\mathbf{r}_2}=0100
$.
Then
$
\tau(\mathbf{r}_1)
=
1000+\operatorname{tor}(\mathcal{C})$,
$\tau(\mathbf{r}_2)
=
0100+\operatorname{tor}(\mathcal{C})
$.
Since
$
1000+0100=1100\in \operatorname{tor}(\mathcal{C})
$,
these two cosets are equal. Moreover,
$
\tau(1111)
=
\tau(\mathbf{r}_1+\mathbf{r}_2)
=
1100+\operatorname{tor}(\mathcal{C})
=
\operatorname{tor}(\mathcal{C})
$.
%
Now, consider the following:
\begin{enumerate}

\item[(i)]
For $\mathbf r=0000$, we have
$
\tau(\sigma(0000))
=
\tau(0000)
=
\operatorname{tor}(\mathcal{C})
$
and
$
\sigma(\tau(0000))
=
\sigma(\operatorname{tor}(\mathcal{C}))
=
\operatorname{tor}(\mathcal{C})
$.
Hence,
$
\tau(\sigma(0000))=\sigma(\tau(0000))
$.

\item[(ii)]
For $\mathbf r=1010=\mathbf{r}_1$, we have
$
\sigma(\mathbf{r}_1)=\mathbf{r}_2
$.
Thus,
$
\tau(\sigma(\mathbf{r}_1))
=
\tau(\mathbf{r}_2)
=
0100+\operatorname{tor}(\mathcal{C})
$.
On the other hand,
$
\sigma(\tau(\mathbf{r}_1))
=
\sigma(1000+\operatorname{tor}(\mathcal{C}))
=
0001+\operatorname{tor}(\mathcal{C})
$.
Since
$
0100+0001=0101\in\operatorname{tor}(\mathcal{C})
$,
we get
$
0100+\operatorname{tor}(\mathcal{C})
=
0001+\operatorname{tor}(\mathcal{C})
$.
Hence,
$
\tau(\sigma(\mathbf{r}_1))=\sigma(\tau(\mathbf{r}_1))
$.

\item[(iii)]
For $\mathbf r=0101=\mathbf{r}_2$, we have
$
\sigma(\mathbf{r}_2)=\mathbf{r}_1
$.
Thus,
$
\tau(\sigma(\mathbf{r}_2))
=
\tau(\mathbf{r}_1)
=
1000+\operatorname{tor}(\mathcal{C})
$.
Also,
$
\sigma(\tau(\mathbf{r}_2))
=
\sigma(0100+\operatorname{tor}(\mathcal{C}))
=
0010+\operatorname{tor}(\mathcal{C})
$.
Since
$
1000+0010=1010\in\operatorname{tor}(\mathcal{C})
$,
we obtain
$
1000+\operatorname{tor}(\mathcal{C})
=
0010+\operatorname{tor}(\mathcal{C})
$.
Hence,
$
\tau(\sigma(\mathbf{r}_2))=\sigma(\tau(\mathbf{r}_2))
$.

\item[(iv)]
For $\mathbf r=1111=\mathbf{r}_1+\mathbf{r}_2$, we have
$
\tau(1111)
=
\operatorname{tor}(\mathcal{C})
$.
Since $\sigma(1111)=1111$, it follows that
$
\tau(\sigma(1111))
=
\operatorname{tor}(\mathcal{C})
$.
Also, since $\operatorname{tor}(\mathcal{C})$ is cyclic,
$
\sigma(\tau(1111))
=
\sigma(\operatorname{tor}(\mathcal{C}))
=
\operatorname{tor}(\mathcal{C})
$.
Therefore,
$
\tau(\sigma(1111))=\sigma(\tau(1111))
$.

\end{enumerate}

Therefore,
$
\tau(\sigma(\mathbf r))
=
\sigma(\tau(\mathbf r))
$
for all
$
\mathbf r\in\operatorname{res}(\mathcal{C}).
$
By Theorem~\ref{thm:cyclic code characterization}, $\mathcal{C}$ is a cyclic $I_2$-code. Moreover, since
$
1000\notin \operatorname{tor}(\mathcal{C})
$,
we have
$
\tau(\mathbf{r}_1)\neq \operatorname{tor}(\mathcal{C})
$.
Thus, $\mathcal{C}$ is twisted.
\end{example}


If $\mathcal{C}$ is an untwisted $I_2$-code, then for all $\mathbf{r} \in \operatorname{res} ( \mathcal{C} )$, $\tau ( \mathbf{r} ) = \operatorname{tor} ( \mathcal{C} )$, and so $\tau ( \sigma ( \mathbf{r} ) ) = \operatorname{tor} ( \mathcal{C} ) = \sigma ( \tau ( \mathbf{r} ) )$ since $\operatorname{res} ( \mathcal{C} )$ and $\operatorname{tor} ( \mathcal{C} )$ are both cyclic. Thus, we obtain the following corollary from the previous theorem.

\begin{corollary}  \label{cor:char-untwisted}
Let $\mathcal{C}$ be an untwisted $I_2$-code. Then $ \mathcal{C}$ is a cyclic code over $I_2$ if and only if $\operatorname{res}(\mathcal{C})$ and $\operatorname{tor} (\mathcal{C}) $ are both binary cyclic codes.
\end{corollary}


\begin{theorem}[Theorem 3 in \cite{Duality_I2E2H2}]  \label{thm:resC and torC of Cdual}
If $\mathcal{C}$ is an $I_2$-code, then $\operatorname{res} ( \mathcal{C}^{\perp} ) = \operatorname{res} ( \mathcal{C} )^{\perp}$ and $\operatorname{tor}  ( \mathcal{C}^{\perp} ) = \mathbb{F}_{2}^{n}$.
\end{theorem}

\begin{theorem}[Theorem 4 in \cite{Duality_I2E2H2}]  \label{thm:Cdual as direct sum}
If $\mathcal{C}$ is an $I_2$-code, then $\mathcal{C}^\perp = \mathtt{a}  \operatorname{res} (\mathcal{C})^\perp  \oplus  \mathtt{b} \mathbb{F}_2^n$.
\end{theorem}

\begin{proposition}  \label{thm:dual-is-also-cyclic}
If $\mathcal{C}$ is a cyclic code over $I_2$, then the dual code $\mathcal{C}^\perp$ is also cyclic.
\end{proposition}

\begin{proof}
Assume that $\mathcal{C}$ is cyclic. From Theorems \ref{thm:resC and torC of Cdual} and \ref{thm:Cdual as direct sum}, we have $\operatorname{res}(\mathcal{C}^\perp) = \operatorname{res}(\mathcal{C})^\perp$, $\operatorname{tor}(\mathcal{C}^\perp) = \mathbb{F}_2^n$, and $\mathcal{C}^{\perp} = \mathtt{a}\operatorname{res}(\mathcal{C})^\perp \oplus \mathtt{b}\mathbb{F}_2^n$. Hence, $\mathcal{C}^\perp$ is untwisted. Since the dual of a binary cyclic code is cyclic, $\operatorname{res}(\mathcal{C})^\perp$ is cyclic, and $\mathbb{F}_2^n$ is trivially cyclic. Thus, by \Cref{cor:char-untwisted}, $\mathcal{C}^\perp$ is cyclic.
\end{proof}

\begin{lemma}  \label{lem:phi-commutes}
Let $\sigma$ be the cyclic shift. Then for all $\mathbf{c} \in I_2^n$, $\phi (\sigma(\mathbf{c})) = \sigma (\phi (\mathbf{c}))$.
\end{lemma}

\begin{proof}
Let $\mathbf{c} = (c_0, c_1, \dots, c_{n-1}) \in I_2^n$. Then $\sigma(\mathbf{c}) = (c_{n-1},c_0,\dots,c_{n-2})$. Applying $\phi$ coordinatewise gives
\begin{align*}
\phi (\sigma(\mathbf{c}))
&= (\phi (c_{n-1}), \phi (c_0),\dots,\phi(c_{n-2}))  \\
&= \sigma (\phi(c_0),\dots,\phi(c_{n-2}),\phi(c_{n-1}))  \\
&= \sigma (\phi (\mathbf{c})).  \tag*{\qedhere}
\end{align*}
\renewcommand{\qedsymbol}{}
\end{proof}

\begin{proposition}  \label{thm:add-code-is-also-cyclic}
Let $\mathcal{C}$ be an $I_2$-code. Then $\mathcal{C}$ is a cyclic code if and only if the additive code $\phi ( \mathcal{C} )$ is also cyclic.
\end{proposition}

\begin{proof}
Assume that $\mathcal{C}$ is cyclic, i.e., $\sigma(\mathcal{C}) = \mathcal{C}$. Let $\mathbf{u} \in \phi(\mathcal{C})$. Then $\mathbf{u} = \phi(\mathbf{c})$ for some $\mathbf{c} \in \mathcal{C}$. By  \Cref{lem:phi-commutes}, $\sigma(\mathbf{u}) = \phi(\sigma(\mathbf{c})) \in \phi(\mathcal{C})$, and so $\phi(\mathcal{C})$ is cyclic.  Conversely, assume that $\phi(\mathcal{C})$ is cyclic. If $\mathbf{c} \in \mathcal{C}$, then $\phi(\mathbf{c}) \in \phi(\mathcal{C})$, so $\sigma (\phi(\mathbf{c})) \in \phi(\mathcal{C})$. By \Cref{lem:phi-commutes}, $\sigma (\phi (\mathbf{c})) = \phi(\sigma(\mathbf{c}))$. Since $\phi$ is bijective coordinatewise, $\sigma(\mathbf{c}) \in \mathcal{C}$. Thus $\mathcal{C}$ is cyclic.
\end{proof}

\begin{proposition}
Let $\mathcal{C}$ be an $I_2$-code. If the additive code $\phi ( \mathcal{C} )$ is cyclic, then $\phi ( \mathcal{C}^\perp )$ is also cyclic.
\end{proposition}

\begin{proof}
This is immediate from Propositions \ref{thm:dual-is-also-cyclic} and \ref{thm:add-code-is-also-cyclic}.
\end{proof}

\begin{corollary}
An  $I_2$-code $\mathcal{C} \subseteq  I_2^n$ is a self-orthogonal cyclic code if and only if $\operatorname{res}(\mathcal{C}) $ is a binary self-orthogonal cyclic code.
\end{corollary}

\begin{proof}
The proof follows from Theorem \ref{thm:cyclic code characterization} and Corollary 1 in \cite{Duality_I2E2H2}.
\end{proof}

\begin{corollary}
A code $\mathcal{C}= \mathtt{a} \mathcal{C}_1 + \mathtt{b} \mathcal{C}_2$ is a quasi self-dual cyclic code over $I_2$ if and only if $ \mathcal{C}_ 1$ is a self-orthogonal cyclic $[n , k]$-code over $\mathbb{F}_2$ and $ \mathcal{C}_2$ is a cyclic $[n, n-k]$-code over $  \mathbb{F}  _2$ such that $ \mathcal{C}_1 \subseteq\mathcal{C}_2$.
\end{corollary}

\begin{proof}
The proof follows from Theorem \ref{thm:cyclic code characterization} and Corollary 4 in \cite{MassFormula_Ip} by taking $p=2$, $\mathcal{C}_1 = \operatorname{res}(\mathcal{C})$, and $\mathcal{C}_ 2 =  \operatorname{tor}  (\mathcal{C})$.
\end{proof}

\begin{corollary}
An $I_2$-code $\mathcal{C}$ of length $n$ is a self-dual cyclic code if and only if $\operatorname{res}(\mathcal{C})$ is a binary self-dual cyclic code and $ \operatorname{tor} (\mathcal{C}) = \mathbb{F} _2^n$.
\end{corollary}

\begin{proof}
The proof is immediate from Theorem \ref{thm:cyclic code characterization} and Theorem 5 in \cite{Duality_I2E2H2}.
\end{proof}

\begin{proposition}\label{thm:dual-containing-I2}
Let $\mathcal{C}$ be an $I_2$-code of length $n$. Then $\mathcal{C}^\perp \subseteq \mathcal{C}$ if and only if $\operatorname{res}(\mathcal{C})^\perp \subseteq \operatorname{res}(\mathcal{C})$ and $\operatorname{tor}(\mathcal{C}) = \mathbb{F}_2^n$.
\end{proposition}

\begin{proof}
Suppose $\mathcal{C}^\perp \subseteq \mathcal{C}$. Let $\mathbf{u} \in \operatorname{res}(\mathcal{C})^\perp$ and $\mathbf{v} \in \mathbb{F}_2^n$. Then $\mathtt{a}\mathbf{u} + \mathtt{b}\mathbf{v} \in \mathcal{C}^\perp \subseteq \mathcal{C}$. Thus, there exist $\mathbf{r} \in \operatorname{res}(\mathcal{C})$ and $\mathbf{t} \in \operatorname{tor} (\mathcal{C})$ such that $\mathtt{a}\mathbf{u} + \mathtt{b}\mathbf{v}
= \mathtt{a}\mathbf{r} + \mathtt{b}(\mathbf{s}_{\mathbf{r}} + \mathbf{t})$. Since representations in $I_2$ are unique coefficientwise, we obtain $\mathbf{u} = \mathbf{r}$ and $\mathbf{v} = \mathbf{s}_{\mathbf{r}} + \mathbf{t}$. Thus, $\mathbf{u} \in \operatorname{res}(\mathcal{C})$, and so $\operatorname{res}(\mathcal{C})^\perp \subseteq \operatorname{res}(\mathcal{C})$. Moreover, the equality $\mathbf{v} = \mathbf{s}_{\mathbf{u}} + \mathbf{t}$ holds for every $\mathbf{v} \in \mathbb{F}_2^n$, implying that $\mathbf{s}_{\mathbf{u}} + \operatorname{tor}(\mathcal{C}) = \mathbb{F}_2^n$. Hence, $\operatorname{tor}(\mathcal{C}) = \mathbb{F}_2^n$.

Conversely, suppose $\operatorname{res}(\mathcal{C})^\perp \subseteq \operatorname{res}(\mathcal{C})$ and $\operatorname{tor}(\mathcal{C}) = \mathbb{F}_2^n$. Let $\mathtt{a}\mathbf{u} + \mathtt{b}\mathbf{v} \in \mathcal{C}^\perp$. Then $\mathbf{u} \in \operatorname{res}(\mathcal{C})^\perp \subseteq \operatorname{res}(\mathcal{C})$. Thus, there exists $\mathbf{t} \in \operatorname{tor}(\mathcal{C})$ such that $\mathtt{a}\mathbf{u} + \mathtt{b}(\mathbf{s}_{\mathbf{u}} + \mathbf{t}) \in \mathcal{C}$. Since $\operatorname{tor}(\mathcal{C}) = \mathbb{F}_2^n$, choose $\mathbf{t} = \mathbf{v}-\mathbf{s}_{\mathbf{u}}$, giving $\mathtt{a}\mathbf{u} + \mathtt{b}\mathbf{v} \in \mathcal{C}$. Therefore, $\mathcal{C}^\perp \subseteq \mathcal{C}$.
\end{proof}

\section{Gray Images  of $\boldsymbol{I_2}$-Codes}
\label{sec: Gray images of I2-codes}

For $r = x \mathtt{a} + y \mathtt{b} \in  I_2$ where $x , y \in \mathbb{F}_2$, define the Gray map $\Phi \colon I_2 \to \mathbb{F}_2^2$ by $\Phi ( r  )  \coloneqq ( x , x + y )$. It can be seen that the map $\Phi$ is a bijection. We extend $\Phi$ naturally from $I_2^n$ to $\mathbb{F}_2^{2n}$, i.e., for any $\mathbf{r} = (r_0 , r_1 , \dotsc , r_{n-1}) \in I_2^n$ where $r_i =  \mathtt{a}  u_i +  \mathtt{b}  v_i$, $0 \leq i \leq n-1$, define $\Phi (\mathbf{r}) = (u (\mathbf{r}) \, | \, u (\mathbf{r}) + v (\mathbf{r}))$ where $u (\mathbf{r}) = (u_0 , u_1 , \dotsc , u_{n-1})$, $v (\mathbf{r}) = (v_0 , v_1 , \dotsc , v_{n-1})$, and ``$|$'' denotes the vector concatenation. It can be checked that the map $\Phi$ is $\mathbb{F}_2$-linear and is a weight-preserving map from $(I_2^n, \operatorname{wt}_L)$ to $(\mathbb{F}_2^{2n},  \operatorname{wt}_H)$, i.e., $\operatorname{wt}_L ( \mathbf{x}  )  =  \operatorname{wt}_H  ( \Phi ( \mathbf{x} )  )$.

We are now ready to characterize the Gray image of an $I_2$-code.

\begin{theorem}   \label{thm:GrayImageGeneral}
Let $\mathcal{C} \subseteq I_2^n$ be an $I_2$-code. For each
$\mathbf{r} \in \operatorname{res}(\mathcal{C})$, choose a representative
$\mathbf{s}_{\mathbf{r}} \in \mathbb{F}_2^n$ of the coset
$\tau (\mathbf{r}) \in \mathbb{F}_2^n / \operatorname{tor}(\mathcal{C})$. Then the Gray image of
$\mathcal{C}$ is given by
$
\Phi(\mathcal{C})
=
\{
(\mathbf{r} \, | \,  \mathbf{r} + \mathbf{s}_{\mathbf{r}} + \mathbf{t})
\mathrel{:}
\mathbf{r} \in \operatorname{res}(\mathcal{C}), \
\mathbf{t} \in \operatorname{tor}(\mathcal{C})
\}
$.
Moreover, this expression is independent of the choice of representatives $\mathbf{s}_{\mathbf{r}}$.
\end{theorem}

\begin{proof}
Recall that for every codeword $\mathbf{c} \in \mathcal{C}$ can be written uniquely as
$
\mathbf{c}
=
\mathtt{a}\mathbf{r} + \mathtt{b}(\mathbf{s}_{\mathbf{r}} + \mathbf{t})
$
where $\mathbf{r} \in \operatorname{res}(\mathcal{C})$ and $\mathbf{t} \in \operatorname{tor}(\mathcal{C})$. Applying the Gray map, we obtain
$
\Phi(\mathbf{c})
=
(\mathbf{r} \, | \,  \mathbf{r} + \mathbf{s}_{\mathbf{r}} + \mathbf{t}),
$
establishing
$
\Phi(\mathcal{C})
\subseteq
\{
(\mathbf{r} \, | \, \mathbf{r} + \mathbf{s}_{\mathbf{r}} + \mathbf{t})  \mathrel{:}
\mathbf{r} \in \operatorname{res}(\mathcal{C}),\
\mathbf{t} \in \operatorname{tor}(\mathcal{C})
\}
$.

Conversely, for any $(\mathbf{r} \, | \, \mathbf{r} + \mathbf{s}_{\mathbf{r}} + \mathbf{t})  \in
\{
(\mathbf{r} \, | \,  \mathbf{r} + \mathbf{s}_{\mathbf{r}} + \mathbf{t})  \mathrel{:}
\mathbf{r} \in \operatorname{res}(\mathcal{C}),\
\mathbf{t} \in \operatorname{tor}(\mathcal{C})
\}$, we have $\mathtt{a} \mathbf{r} + \mathtt{b} {\mathbf{s_r}} \in \mathcal{C}$ and $\mathtt{b} \mathbf{t} \in \mathcal{C}$. Because $\mathcal{C}$ is closed under addition, $\mathtt{a} \mathbf{r} + \mathtt{b} {\mathbf{s_r}} + \mathtt{b} \mathbf{t} \in \mathcal{C}$. Hence,  $\mathtt{a}\mathbf{r} + \mathtt{b} (\mathbf{s}_{\mathbf{r}} + \mathbf{t}) \in \mathcal{C}
$ and $
(\mathbf{r} \, | \, \mathbf{r} + \mathbf{s}_{\mathbf{r}} + \mathbf{t})  = \Phi (  \mathtt{a}\mathbf{r} + \mathtt{b} (\mathbf{s}_{\mathbf{r}} + \mathbf{t})  ) \in \Phi ( \mathcal{C} )$.
Therefore,
$
\{
(\mathbf{r} \, | \,  \mathbf{r} + \mathbf{s}_{\mathbf{r}} + \mathbf{t})  \mathrel{:}
\mathbf{r} \in \operatorname{res}(\mathcal{C}),\
\mathbf{t} \in \operatorname{tor}(\mathcal{C})
\}
\subseteq
\Phi ( \mathcal{C} )
$.
Finally, if another system of representatives satisfies
$\mathbf{s}'_{\mathbf{r}} = \mathbf{s}_{\mathbf{r}} + \mathbf{t}_{\mathbf{r}}$
for some $\mathbf{t}_{\mathbf{r}} \in \operatorname{tor}(\mathcal{C})$, then
\begin{align*}
(\mathbf{r} \, | \, \mathbf{r} + \mathbf{s}'_{\mathbf{r}} + \mathbf{t})
=
(\mathbf{r} \, | \, \mathbf{r} + \mathbf{s}_{\mathbf{r}} + (\mathbf{t} + \mathbf{t}_{\mathbf{r}})),
\end{align*}
and since $\mathbf{t}$ ranges over $\operatorname{tor}(\mathcal{C})$, the resulting set remains unchanged.
\end{proof}

\begin{corollary}  \label{cor:Gray image of dual of C}
Let $\mathcal{C}$ be an $I_2$-code of length $n$.  Then
$
\Phi(\mathcal{C}^\perp) = \{(\mathbf{r} \, | \,   \mathbf{r}  +  \mathbf{t}) \mathrel{:}   \mathbf{r}  \in\operatorname{res} ( \mathcal{C} )^\perp, \,  \mathbf{t}  \in  \mathbb{F} _2^n \}
$.
\end{corollary}

\begin{proof}
From Theorems \ref{thm:resC and torC of Cdual} and \ref{thm:Cdual as direct sum}, it follows that $\mathcal{C}^\perp$ is an untwisted $I_2$-code. Thus, we can set $\mathbf{s}_{\mathbf{r}} = \mathbf{0}$ for all $\mathbf{r} \in \operatorname{res} ( \mathcal{C}^\perp )$, and the result follows from \Cref{thm:GrayImageGeneral}.
\end{proof}

%

\begin{theorem}   \label{thm:dual-phi}
Let $\mathcal{C} \subseteq I_2^n$ be an $I_2$-code. For each
$\mathbf{r} \in \operatorname{res}(\mathcal{C})$, choose a representative
$\mathbf{s}_{\mathbf{r}} \in \mathbb{F}_2^n$ of the coset
$\tau (\mathbf{r}) \in \mathbb{F}_2^n / \operatorname{tor}(\mathcal{C})$. Then the dual of the Gray image of $\mathcal{C}$ is given by
$
\Phi(\mathcal{C})^\perp
=
\{
(\mathbf{u} \, | \, \mathbf{v})    \mathrel{:}
\mathbf{v} \in \operatorname{tor} ( \mathcal{C} )^\perp
,  \,
(\mathbf{u} + \mathbf{v}) \cdot \mathbf{r}
=
\mathbf{v} \cdot \mathbf{s}_{\mathbf{r}}  \
\forall \, \mathbf{r} \in \operatorname{res} ( \mathcal{C} )
\}
$.
Moreover, this expression is independent of the choice of representatives $\mathbf{s}_{\mathbf{r}}$.
\end{theorem}

\begin{proof}
From \Cref{thm:GrayImageGeneral}, we have
\begin{align*}
\Phi(\mathcal{C})
=
\{
(\mathbf{r} \, |  \,  \mathbf{r} + \mathbf{s}_{\mathbf{r}} + \mathbf{t})
\mathrel{:}
\mathbf{r} \in \operatorname{res}(\mathcal{C})  ,  \
\mathbf{t} \in \operatorname{tor}(\mathcal{C})
\}  .
\end{align*}
Let $(\mathbf{u} \, | \, \mathbf{v}) \in \mathbb{F}_2^{2n}$. Then
$(\mathbf{u} \, | \, \mathbf{v}) \in \Phi(\mathcal{C})^\perp$ if and only if
$
(\mathbf{u} \, | \,  \mathbf{v})
\cdot
(\mathbf{r} \, |  \,  \mathbf{r} + \mathbf{s}_{\mathbf{r}} + \mathbf{t})
= 0
$
for all $\mathbf{r} \in \operatorname{res}(\mathcal{C})$ and $\mathbf{t} \in \operatorname{tor}(\mathcal{C})$. Computing the inner product
gives
\begin{align*}
0
=
(\mathbf{u} \, | \,  \mathbf{v})
\cdot
(\mathbf{r} \, |  \,  \mathbf{r} + \mathbf{s}_{\mathbf{r}} + \mathbf{t})    
=
\mathbf{u} \cdot \mathbf{r}
+
\mathbf{v} \cdot (\mathbf{r} + \mathbf{s}_{\mathbf{r}} + \mathbf{t})   
=
(\mathbf{u} + \mathbf{v}) \cdot \mathbf{r}
+
\mathbf{v} \cdot \mathbf{s}_{\mathbf{r}}
+
\mathbf{v} \cdot \mathbf{t}
\end{align*}
for all $\mathbf{r} \in \operatorname{res}(\mathcal{C})$ and $\mathbf{t} \in \operatorname{tor}(\mathcal{C})$.

First, fix $\mathbf{r} \in \operatorname{res}(\mathcal{C})$ and let $\mathbf{t}$ vary over $\operatorname{tor}(\mathcal{C})$. The above
expression must vanish for all $\mathbf{t} \in \operatorname{tor}(\mathcal{C})$, so we must have
$\mathbf{v} \cdot \mathbf{t} = 0$ for all $\mathbf{t} \in \operatorname{tor}(\mathcal{C})$, that is,
$
\mathbf{v} \in \operatorname{tor}(\mathcal{C})^\perp.
$
With this condition imposed, the orthogonality condition simplifies to
$
(\mathbf{u} + \mathbf{v}) \cdot \mathbf{r}
+
\mathbf{v} \cdot \mathbf{s}_{\mathbf{r}}
=
0
$
for all $\mathbf{r} \in \operatorname{res}(\mathcal{C})$,
or equivalently,
$
(\mathbf{u} + \mathbf{v}) \cdot \mathbf{r}
=
\mathbf{v} \cdot \mathbf{s}_{\mathbf{r}}
$
for all $\mathbf{r} \in \operatorname{res}(\mathcal{C})$.
Thus, $(\mathbf{u} \, | \, \mathbf{v}) \in \Phi(\mathcal{C})^\perp$ if and only if
$\mathbf{v} \in \operatorname{tor}(\mathcal{C})^\perp$ and
$
(\mathbf{u} + \mathbf{v}) \cdot \mathbf{r}
=
\mathbf{v} \cdot \mathbf{s}_{\mathbf{r}}
$
for all  $\mathbf{r} \in \operatorname{res}(\mathcal{C})
$,
which yields the claimed description of $\Phi(\mathcal{C})^\perp$.

Finally, we check that this description is independent of the choice of
representatives $\mathbf{s}_{\mathbf{r}}$. Suppose we replace each
$\mathbf{s}_{\mathbf{r}}$ by
$
\mathbf{s}'_{\mathbf{r}} = \mathbf{s}_{\mathbf{r}} + \mathbf{t}_{\mathbf{r}}
$
where $\mathbf{t}_{\mathbf{r}} \in \operatorname{tor}(\mathcal{C})$.
For $\mathbf{v} \in \operatorname{tor}(\mathcal{C})^\perp$, we have
\begin{align*}
\mathbf{v} \cdot \mathbf{s}'_{\mathbf{r}}
=
\mathbf{v} \cdot \mathbf{s}_{\mathbf{r}}
+
\mathbf{v} \cdot \mathbf{t}_{\mathbf{r}}
=
\mathbf{v} \cdot \mathbf{s}_{\mathbf{r}},
\end{align*}
since $\mathbf{v} \cdot \mathbf{t}_{\mathbf{r}} = 0$ for all
$\mathbf{t}_{\mathbf{r}} \in \operatorname{tor}(\mathcal{C})$. Hence, the condition
$
(\mathbf{u} + \mathbf{v}) \cdot \mathbf{r}
=
\mathbf{v} \cdot \mathbf{s}_{\mathbf{r}}
$
for all $\mathbf{r} \in \operatorname{res}(\mathcal{C})$
is equivalent to
$
(\mathbf{u} + \mathbf{v}) \cdot \mathbf{r}
=
\mathbf{v} \cdot \mathbf{s}'_{\mathbf{r}}
$
for all $\mathbf{r} \in \operatorname{res}(\mathcal{C})
$,
so the resulting set $\Phi(\mathcal{C})^\perp$ does not depend on the chosen
system of representatives.
\end{proof}

\begin{proposition}  \label{thm:resCdual subset of resC}
Let $\mathcal{C}$ be an untwisted $I_2$-code of length $n$.  Then $\Phi(\mathcal{C})^\perp \subseteq \Phi(\mathcal{C})$ if and only if $\operatorname{res} (\mathcal{C})^\perp \subseteq \operatorname{tor} (\mathcal{C})$ and $\operatorname{res} (\mathcal{C})^\perp + \operatorname{tor} (\mathcal{C})^\perp \subseteq \operatorname{res} (\mathcal{C})$.
\end{proposition}

\begin{proof}
Let $\mathcal{C}$ be an untwisted $I_2$-code. Then we can set $\mathbf{s}_{\mathbf{r}} = \mathbf{0}$ for all $\mathbf{x} \in \operatorname{res} ( \mathcal{C} )$. Thus, from Theorems \ref{thm:GrayImageGeneral} and \ref{thm:dual-phi}, we obtain
\begin{align*}
\Phi(\mathcal{C})
&=
\{ (\mathbf{r} \, | \, \mathbf{r} + \mathbf{t}) \mathrel{:}  \mathbf{r} \in \operatorname{res}(\mathcal{C}),\ \mathbf{t} \in \operatorname{tor}(\mathcal{C}) \}
\end{align*}
and
\begin{align*}
\Phi(\mathcal{C})^\perp
&=
\{ (\mathbf{u} \, | \, \mathbf{v}) \mathrel{:}  \mathbf{v} \in \operatorname{tor}(\mathcal{C})^\perp,\ \mathbf{u} + \mathbf{v} \in \operatorname{res}(\mathcal{C})^\perp \}.
\end{align*}

Assume $\Phi(\mathcal{C})^\perp \subseteq \Phi(\mathcal{C})$.
Let $\mathbf{u} \in \operatorname{res}(\mathcal{C})^\perp$ and $\mathbf{v} \in \operatorname{tor}(\mathcal{C})^\perp$. Then
$
(\mathbf{u} + \mathbf{v} \mid \mathbf{v}) \in \Phi(\mathcal{C})^\perp
\subseteq \Phi(\mathcal{C})
$,
so there exists $\mathbf{r} \in \operatorname{res}(\mathcal{C})$ and $\mathbf{t} \in \operatorname{tor}(\mathcal{C})$ such that
$
(\mathbf{u} + \mathbf{v} \, | \, \mathbf{v})
= (\mathbf{r} \, | \, \mathbf{r} + \mathbf{t})
$.
Hence,
$
\mathbf{r} = \mathbf{u} + \mathbf{v}
$,
$\mathbf{v} = \mathbf{r} + \mathbf{t} = \mathbf{u} + \mathbf{v} + \mathbf{t}
$,
and so
$
\mathbf{t} = \mathbf{u}
$.
Thus, $\mathbf{u} \in \operatorname{tor}(\mathcal{C})$, and since $\mathbf{r} = \mathbf{u} + \mathbf{v} \in \operatorname{res}(\mathcal{C})$ for all $\mathbf{u} \in \operatorname{res}(\mathcal{C})^\perp$ and $\mathbf{v} \in \operatorname{tor}(\mathcal{C})^\perp$, we have
$
\operatorname{res}(\mathcal{C})^\perp  \subseteq \operatorname{tor}(\mathcal{C})$
and
$
\operatorname{res}(\mathcal{C})^\perp + \operatorname{tor}(\mathcal{C})^\perp \subseteq \operatorname{res}(\mathcal{C})
$.

Conversely, assume that
$
\operatorname{res}(\mathcal{C})^\perp \subseteq \operatorname{tor}(\mathcal{C})
$
and
$
\operatorname{res}(\mathcal{C})^\perp + \operatorname{tor}(\mathcal{C})^\perp \subseteq \operatorname{res}(\mathcal{C})
$.
Let $(\mathbf{u} \, | \, \mathbf{v}) \in \Phi(\mathcal{C})^\perp$. Then $\mathbf{v} \in \operatorname{tor}(\mathcal{C})^\perp$ and
$
\mathbf{y} := \mathbf{u} + \mathbf{v} \in \operatorname{res}(\mathcal{C})^\perp
$.
By the second inclusion, we have
$
\mathbf{r} := \mathbf{y} + \mathbf{v} \in \operatorname{res}(\mathcal{C})
$,
and by the first inclusion, we have $\mathbf{y} \in \operatorname{tor}(\mathcal{C})$. Thus, if we set $\mathbf{t} := \mathbf{y} \in \operatorname{tor}(\mathcal{C})$, then
$
(\mathbf{u} \, | \,  \mathbf{v})
= (\mathbf{y} + \mathbf{v} \, | \,  \mathbf{v})
= (\mathbf{r} \, | \,  \mathbf{r} + \mathbf{t}) \in \Phi(\mathcal{C})
$.
Therefore, $\Phi(\mathcal{C})^\perp \subseteq \Phi(\mathcal{C})$.
\end{proof}

\begin{proposition}
Let $\mathcal{C}$ be an $I_2$-code of length $n$. Then $ \Phi (\mathcal{C}^\perp ) = \Phi (\mathcal{C})^\perp$ if and only if  $\mathcal{C} =  \{ \mathbf{0}  \}$.
\end{proposition}

\begin{proof}
Suppose $\Phi(\mathcal{C})^\perp = \Phi(\mathcal{C}^\perp)$.
From \Cref{cor:Gray image of dual of C}, we have
\begin{align*}
\Phi(\mathcal{C}^\perp) = \{(\mathbf{r} \, | \,   \mathbf{r}  +  \mathbf{t}) \mathrel{:}   \mathbf{r}  \in\operatorname{res} ( \mathcal{C} )^\perp, \,  \mathbf{t}  \in  \mathbb{F} _2^n \}  .
\end{align*}
Fix an arbitrary vector $\mathbf{v} \in \mathbb{F}_2^n$. Taking
$\mathbf{r} = \mathbf{0}$ and $\mathbf{t} = \mathbf{v}$, we have
$
(\mathbf{0} \mid \mathbf{v}) \in \Phi(\mathcal{C}^\perp)
$.
By the hypothesis, this implies that
$
(\mathbf{0} \mid \mathbf{v}) \in \Phi(\mathcal{C})^\perp
$.
From \Cref{thm:dual-phi}, every $(\mathbf{u} \mid \mathbf{v}) \in \Phi(\mathcal{C})^\perp$ satisfies
$
\mathbf{v}  \in \operatorname{tor}(\mathcal{C})^\perp
$
and
$
(\mathbf{u} + \mathbf{v}) \cdot \mathbf{r}
= \mathbf{v} \cdot \mathbf{s}_{\mathbf{r}}
$
for all $\mathbf{r} \in \operatorname{res}(\mathcal{C})$.
In particular, since $(\mathbf{0} \mid \mathbf{v}) \in \Phi(\mathcal{C})^\perp$,
we obtain
$
\mathbf{v} \in \operatorname{tor}(\mathcal{C})^\perp
$.
Because $\mathbf{v} \in \mathbb{F}_2^n$ is arbitrary, it follows that
$
\operatorname{tor}(\mathcal{C})^\perp = \mathbb{F}_2^n
$,
and hence,
$
\operatorname{tor}(\mathcal{C}) = \{\mathbf{0}\}
$.
Since $\operatorname{res}(\mathcal{C}) \subseteq \operatorname{tor}(\mathcal{C})$,
we also obtain
$
\operatorname{res}(\mathcal{C}) = \{\mathbf{0}\}
$.
From the twist-map description of $\mathcal{C}$, we have
$
\mathcal{C}
=
\left\{
\mathtt{b}\,\mathbf{s}_{\mathbf{0}}
\right\}
$.
Moreover,
$
\tau (\mathbf{0}) = \mathbf{s}_{\mathbf{0}} + \operatorname{tor}(\mathcal{C}) .
$
Since $\operatorname{tor}(\mathcal{C}) = \{\mathbf{0}\}$ and $\alpha (\mathbf{0}) = \mathbf{0}$, it follows that
$
\mathbf{s}_{\mathbf{0}} = \mathbf{0}
$.
Therefore,
$
\mathcal{C} = \{\mathtt{b}\,\mathbf{s}_{\mathbf{0}}\} = \{\mathbf{0}\}
$.

Conversely, if $\mathcal{C} = \{ \mathbf{0} \}$, then
$
\Phi(\mathcal{C}) = \{ \mathbf{0} \}
$,
and so
$
\Phi(\mathcal{C})^\perp = \mathbb{F}_2^{2n}
$.
Moreover, $\mathcal{C}^\perp = I_2^n$, and since $\Phi$ is a bijection from $I_2^n \to \mathbb{F}_2^{2n}$, we have
$
\Phi(\mathcal{C}^\perp) = \mathbb{F}_2^{2n}
$.
Therefore,
$
\Phi(\mathcal{C})^\perp = \Phi(\mathcal{C}^\perp)
$.
\end{proof}

\begin{theorem}  \label{thm:parameters of Gray image}
If $\mathcal{C}$ is an $I_2$-code of length $n$, then the Gray image $\Phi (\mathcal{C})$ has length $2n$ and size $|\mathcal{C}|$. If $\mathcal{C}$ is untwisted, then the minimum distance of $\Phi ( \mathcal{C} )$ is at least $\min\{  2 d (\operatorname{res} ( \mathcal{C} ) )  ,  d (\operatorname{tor} ( \mathcal{C} ) )  \}$, and equality holds if $\operatorname{res} ( \mathcal{C} )$ and $\operatorname{tor} ( \mathcal{C} )$ are both nonzero binary codes.
\end{theorem}

\begin{proof}
Let $ \mathcal{C}$ be an $I_2$-code of length $n$. The length and the size of $\Phi ( \mathcal{C})$ follow directly from the definition and bijectivity of the Gray map $\Phi$. Now, let $\mathcal{C}$ be an untwisted $I_2$-code. Take any nonzero codeword $ \mathbf{c}  \in \mathcal{C}$. Then $ \mathbf{c} =   \mathtt{a}  \mathbf{r}  +  \mathtt{b}  \mathbf{t} $ where $ \mathbf{r}  \in  \operatorname{res} (\mathcal{C}) $ and $ \mathbf{t}  \in \operatorname{tor} (\mathcal{C}) $. From \Cref{thm:GrayImageGeneral}, we have $\Phi( \mathbf{c} ) = \Phi( \mathtt{a}  \mathbf{r}  +  \mathtt{b}  \mathbf{t} ) = ( \mathbf{r} \, | \, \mathbf{r} + \mathbf{t} )$, and so
\begin{align*}
\operatorname{wt}_H ( \Phi( \mathbf{c} ) ) = \operatorname{wt}_H ( \mathbf{r} \, | \, \mathbf{r} + \mathbf{t} ) =  \operatorname{wt}_H ( \mathbf{r} ) + \operatorname{wt}_H ( \mathbf{r} + \mathbf{t} )  .
\end{align*}
\emph{Case 1:} $ \mathbf{r} = \mathbf{0}$, $ \mathbf{t} \neq \mathbf{0}$.
Then $ \mathbf{c} =  \mathtt{b} \mathbf{t} $ and
\begin{align*}
\Phi( \mathbf{c} ) = ( \mathbf{0} \, | \, \mathbf{t}  ),
\qquad
\operatorname{wt}_H ( \Phi ( \mathbf{c} ) )  = \operatorname{wt}_H ( \mathbf{0} \, | \, \mathbf{t}  ) = \operatorname{wt}_H ( \mathbf{t} ) \ge d ( \operatorname{tor} (\mathcal{C}) )  .
\end{align*}

\noindent  \emph{Case 2:} $ \mathbf{r} \neq \mathbf{0} $, $ \mathbf{t} = \mathbf{0} $.
Then $ \mathbf{c} =  \mathtt{a}  \mathbf{r} $ and
\begin{align*}
\Phi ( \mathbf{c} ) = ( \mathbf{r} \, | \, \mathbf{r} ),
\qquad
\operatorname{wt}_H ( \Phi ( \mathbf{c} ) )  =  \operatorname{wt}_H ( \mathbf{r} \, | \, \mathbf{r} ) = \operatorname{wt}_H ( \mathbf{r} ) + \operatorname{wt}_H ( \mathbf{r} ) = 2  \operatorname{wt}_H ( \mathbf{r} ) \ge 2 d( \operatorname{res} (\mathcal{C}) ) .
\end{align*}

\noindent\emph{Case 3:}  $ \mathbf{r} \neq \mathbf{0}$, $ \mathbf{t} \neq \mathbf{0}$.
Let $ \mathbf{r}  =  (r_1 ,r_2 , \dotsc , r_n ) \in \operatorname{res} (\mathcal{C}) $ and $ \mathbf{t} = ( t_1 , t_2 , \dotsc , t_n ) \in \operatorname{tor} (\mathcal{C}) $. For each coordinate $i$ with $t_i = 1$, at least one of $r_i$ or $r_i + t_i$ equals $1$, so the support of $ \mathbf{t} $ is contained in the union of the supports of $ \mathbf{r} $ and $ \mathbf{r}  +  \mathbf{t} $. Hence,
\begin{align*}
\operatorname{wt}_H ( \mathbf{t} )   \le   \operatorname{wt}_H ( \mathbf{r} ) + \operatorname{wt}_H ( \mathbf{r}  + \mathbf{t}  ) = \operatorname{wt}_H  ( \mathbf{r} \, | \,  \mathbf{r}  + \mathbf{t}  )  =  \operatorname{wt}_H ( \Phi ( \mathbf{c} )  )  .
\end{align*}
From Case 1, we know that $\operatorname{wt} ( \mathbf{t} ) \ge d( \operatorname{tor} (\mathcal{C}) )$, and so
\begin{align*}
\operatorname{wt}_H ( \Phi ( \mathbf{c} ) ) = \operatorname{wt}_H  ( \mathbf{r} \, | \, \mathbf{r} + \mathbf{t}  ) \ge \operatorname{wt}_H ( \mathbf{t} ) \ge d ( \operatorname{tor} (\mathcal{C}) ) .
\end{align*}
In all three cases, for any nonzero codeword $ \mathbf{c}  \in  \mathcal{C}$, we have
\begin{align*}
\operatorname{wt}_H ( \Phi ( \mathbf{c} ) )  \geq  \min\{2 d(\operatorname{res}(\mathcal{C})), d(\operatorname{tor}(\mathcal{C}))\}  .
\end{align*}

Every nonzero codeword of $\Phi(\mathcal{C})$ has weight at least $\min  \{ 2 d( \operatorname{res} ( \mathcal{C} ) )  , d( \operatorname{tor} ( \mathcal{C} ) )  \}$, and so we obtain
\begin{align*}
d( \Phi(\mathcal{C} ) )  \ge  \min  \{ 2 d( \operatorname{res} ( \mathcal{C} ) )  , d( \operatorname{tor} ( \mathcal{C} ) )  \}  .
\end{align*}
Now, assume that both $ \operatorname{res} (\mathcal{C}) $ and $ \operatorname{tor} (\mathcal{C}) $ are nonzero binary codes. Then there exist $ \mathbf{r} _0 \in \operatorname{res} (\mathcal{C}) \setminus \{ \mathbf{0} \}$,  $  \mathbf{t}_0 \in \operatorname{tor} (\mathcal{C}) \setminus\{ \mathbf{0} \}$ such that $\operatorname{wt}_H ( \mathbf{r} _0 ) = d ( \operatorname{res} (\mathcal{C}) )$ and $\operatorname{wt}_H (  \mathbf{t}_0 ) = d ( \operatorname{tor} (\mathcal{C}) )$. Consider $ \mathbf{c}_1 = \mathtt{a}  \mathbf{r} _0 \in \mathcal{C}$. Then
\begin{align*}
\Phi( \mathbf{c}_1) = ( \mathbf{r} _0 \, | \, \mathbf{r} _0), \qquad
\operatorname{wt}_H ( \Phi( \mathbf{c}_1 ) ) = 2 \operatorname{wt}_H ( \mathbf{r} _0) = 2 d ( \operatorname{res} (\mathcal{C}) )  ,
\end{align*}
and so
$
d(\Phi(\mathcal{C})) \le 2 d( \operatorname{res} (\mathcal{C}) )
$.
Similarly, consider $ \mathbf{c}_2 = \mathtt{b}  \mathbf{t}_0 \in \mathcal{C}$. Then
\begin{align*}
\Phi( \mathbf{c}_2) = ( \mathbf{0}  \, | \, \mathbf{t}_0 ) , \qquad
\operatorname{wt}_H (\Phi( \mathbf{c}_2)) = \operatorname{wt}_H (  \mathbf{t}_0) = d( \operatorname{tor} (\mathcal{C}) ) ,
\end{align*}
and so
$
d(\Phi(\mathcal{C})) \le d( \operatorname{tor} (\mathcal{C}) )
$.
Combining these,
\begin{align*}
d(\Phi(\mathcal{C})) \leq  \min\{2 d(\operatorname{res}(\mathcal{C})), d(\operatorname{tor}(\mathcal{C}))\}  .
\end{align*}
Together with the previously established lower bound, we conclude that
\begin{align*}
d(\Phi(\mathcal{C})) =  \min\{2 d(\operatorname{res}(\mathcal{C})), d(\operatorname{tor}(\mathcal{C}))\} .  \tag*{\qedhere}
\end{align*}
\renewcommand{\qedsymbol}{}
\end{proof}

For any positive integer $s$, let $\sigma^{\otimes s}$ be the quasi-cyclic shift given by
\begin{align*}
\sigma^{\otimes s} ( \mathbf{x} ^{(1)} \, | \,  \mathbf{x}  ^{(2)} \, | \, \dotsc  \, |   \, \mathbf{x}  ^{(s)}  ) = ( \sigma ( \mathbf{x}  ^{(1)}  )  \, |  \,   \sigma ( \mathbf{x}  ^{(2)}  )  \,  |  \,  \dotsc  \,  |  \,  \sigma ( \mathbf{x}  ^{(s)}  ) )
\end{align*}
where $ \mathbf{x} ^{(1)} ,   \mathbf{x}  ^{(2)} , \dotsc ,  \mathbf{x}  ^{(s)} \in  \mathbb{F} _2^n$. A \textit{quasi-cyclic code} $\mathcal{C}$ of index $s$ and length $ns$ over $ \mathbb{F} _2$ is a subset of $( \mathbb{F} _2^n  )^s$ such that $\sigma^{\otimes s} (\mathcal{C}) = \mathcal{C}$.

\begin{proposition}  \label{prop: Gray map and quasi-cyclic shift}
Let $\Phi$ be the Gray map defined above, $\sigma$ be the cyclic shift, and $\sigma^{\otimes 2}$ be the quasi-cyclic shift on $( \mathbb{F} _2^n )^2$. Then $\Phi ( \sigma (\mathcal{C}) ) = \sigma^{\otimes 2} ( \Phi (\mathcal{C}) )$.
\end{proposition}

\begin{proof}
Let $ \mathbf{c} = (c_0 , c_1 , \dotsc , c_{n-1}) \in I_2^n$ where $c_i =  \mathtt{a}  x_i  +  \mathtt{b}  y_i \in I_2$ and $x_i, y_i \in   \mathbb{F} _2$ for $i = 0 , 1 , \dotsc , n-1$. Then we have
\begin{align*}
\Phi ( \sigma ( \mathbf{c} )) &= \Phi ( c_{n-1} , c_0 , c_1 , \dotsc , c_{n-2}  )  \\
&= ( x_{n-1} , x_0 , x_1 , \dotsc , x_{n-2} , x_{n-1} + y_{n-1} , x_0 + y_0 , x_1 + y_1 , \dotsc , x_{n-2} + y_{n-2}  )
\end{align*}
and
\begin{align*}
\sigma^{\otimes 2} (\Phi ( \mathbf{c} )) &= \sigma^{\otimes 2} (  x_0 , x_1 , \dotsc ,  x_{n-1} , x_0 + y_0 , x_1 + y_1 , \dotsc , x_{n-1} + y_{n-1}   )  \\
&= ( x_{n-1} , x_0 , x_1 , \dotsc , x_{n-2} , x_{n-1} + y_{n-1} , x_0 + y_0 , x_1 + y_1 , \dotsc , x_{n-2} + y_{n-2}  )  .
\end{align*}
Thus, $\Phi ( \sigma (\mathcal{C}) ) = \sigma^{\otimes 2} ( \Phi (\mathcal{C}) )$.
\end{proof}

\begin{theorem}  \label{thm: cyclic code implies Gray image is quasi-cyclic}
If $\mathcal{C}$ is a cyclic code of length $n$ over $I_2$, then $\Phi (\mathcal{C})$ is a binary quasi-cyclic code of index $2$ and length $2n$.
\end{theorem}

\begin{proof}
Since $\mathcal{C}$ is a cyclic code of length $n$ over $I_2$, we have $\sigma (\mathcal{C}) = \mathcal{C}$. Taking $\Phi$ for both sides and using  \Cref{prop: Gray map and quasi-cyclic shift}, we get $\sigma^{\otimes 2} ( \Phi (\mathcal{C}) ) = \Phi (\mathcal{C})$. This implies that $\Phi (\mathcal{C})$ is a binary quasi-cyclic code of index $2$ and length $2n$.
\end{proof}

Let $\Lambda : I_2^n \to  \mathbb{F} _2^{2n}$ be the permuted version of the Gray map $\Phi$, i.e.,
\begin{align*}
\Lambda ( c_0 , c_1 , \dotsc , c_{n-1} ) &= ( \Phi (c_0) , \Phi (c_1 ) , \dotsc , \Phi (c_{n-1} )  )  \\
&= (  x_0 , x_0 + y_0 , x_1 , x_1 + y_1 , \dotsc ,  x_{n-1} ,  x_{n-1} + y_{n-1}   )
\end{align*}
where $c_j =  \mathtt{a}  x_j  +  \mathtt{b}  y_j \in I_2$ and $x_j, y_j \in   \mathbb{F} _2$ for $i = 0 , 1 , \dotsc , n-1$.

\begin{theorem}  \label{thm: Gray map variant and cyclic shift}
For any $ \mathbf{c}  \in  I_2^n$, we have $\Lambda ( \sigma ( \mathbf{c} ) ) = \sigma^2 ( \Lambda ( \mathbf{c} ) )$ where $\Lambda$ is the map defined above and $\sigma$ is the cyclic shift.
\end{theorem}

\begin{proof}
Let $ \mathbf{c} = (c_0 , c_1 , \dotsc , c_{n-1}) \in I_2^n$ where $c_j =  \mathtt{a}  x_j  +  \mathtt{b}  y_j \in I_2$ and $x_j, y_j \in   \mathbb{F} _2$ for $i = 0 , 1 , \dotsc , n-1$. Then we have
\begin{align*}
\Lambda ( \sigma ( \mathbf{c} )) &= \Lambda ( c_{n-1} , c_0 , c_1 , \dotsc , c_{n-2}  )  \\
&= ( x_{n-1} , x_{n-1} + y_{n-1} , x_0, x_0 + y_0,  x_1 , x_1 + y_1 , \dotsc , x_{n-2}  , x_{n-2} + y_{n-2}  )
\end{align*}
and
\begin{align*}
\sigma^2 (\Lambda ( \mathbf{c} )) &= \sigma^2 (  x_0 , x_0 + y_0 , x_1 , x_1 + y_1 , \dotsc ,  x_{n-1} ,  x_{n-1} + y_{n-1}   )  \\
&= ( x_{n-1} , x_{n-1} + y_{n-1} , x_0, x_0 + y_0,  x_1 , x_1 + y_1 , \dotsc , x_{n-2}  , x_{n-2} + y_{n-2}  ) .
\end{align*}
Thus, $\Lambda ( \sigma ( \mathbf{c} ) ) = \sigma^2 ( \Lambda ( \mathbf{c} ) )$.
\end{proof}

\begin{theorem}
If $\mathcal{C}$ is a cyclic code of length $n$ over $I_2$, then $\Lambda(\mathcal{C})$ is equivalent to a binary quasi-cyclic code of index $2$ and length $2n$.
\end{theorem}

\begin{proof}
Since $\mathcal{C}$ is a cyclic code of length $n$ over $I_2$, we have $\sigma (\mathcal{C}) = \mathcal{C}$. Taking $\Lambda$ for both sides and using  \Cref{thm: Gray map variant and cyclic shift}, we get $\sigma^2 ( \Lambda (\mathcal{C}) ) = \Lambda (\mathcal{C})$. This implies that $\Lambda (\mathcal{C})$ is equivalent to a binary quasi-cyclic code of index $2$ and length $2n$.
\end{proof}

\begin{example}
Consider the cyclic code $\mathcal{C}$ of length $4$ over $I_2$ with a generator matrix given by
\begin{align*}
G = \begin{pmatrix*}[c]
\mathtt{a}  &  \mathtt{a}  &  \mathtt{a}  &  \mathtt{a}  \\
\mathtt{b}  &  0  &  \mathtt{b}  &  0   \\
0  &  \mathtt{b}  &  0  &  \mathtt{b}
\end{pmatrix*}
=
\begin{pmatrix*}[c]
\mathtt{a} G_1 \\
\mathtt{b} G_2
\end{pmatrix*} ,
\end{align*}
where $G_1 = \left(
\begin{array}{cccc}
1 & 1 & 1 & 1
\end{array}
\right)$
and
$G_2 = \left(
\begin{array}{cccc}
1 & 0 & 1 & 0  \\
0 & 1 & 0 & 1
\end{array}
\right)$
are generator matrices for $\operatorname{res} ( \mathcal{C} )$ and $\operatorname{tor} ( \mathcal{C} )$, respectively. From \Cref{thm:gen-matrix-cyclic-I2-codes}, we see that $\mathcal{C}$ is untwisted. Moreover, by \Cref{thm: cyclic code implies Gray image is quasi-cyclic}, the Gray image
\begin{align*}
\Phi ( \mathcal{C} )
&=
\{
00000000,
00001010,
00000101,
00001111,
11111111,  \\
&\phantom{====}11110101,
11111010,
11110000
\}
\end{align*}
is a binary quasi-cyclic code of index $2$ and length $8$. Since $\mathcal{C}$ is untwisted and $\operatorname{res} ( \mathcal{C} )$ and $\operatorname{tor} ( \mathcal{C} )$ are both nonzero binary codes, \Cref{thm:parameters of Gray image} tells us that
\begin{align*}
d ( \Phi ( \mathcal{C}) ) = 2 = \min \{  2 d ( \operatorname{res} ( \mathcal{C} ) ) , d ( \operatorname{tor} ( \mathcal{C}) )  \}
\end{align*}
where $d ( \operatorname{res} ( \mathcal{C} ) ) = 4$ and $d ( \operatorname{tor} ( \mathcal{C} ) ) = 2$.
\end{example}

\section{Computational Results}
\label{sec: comp results}


In this section, we present the classification of nonzero cyclic codes over $I_2$ for lengths $n \le 7$ using \texttt{MAGMA} \cite{MAGMA}.
Table~\ref{tbl: nonzero cyclic codes over I2 summary} presents the number of permutation-inequivalent nonzero cyclic codes over $I_2$ and the highest possible minimum distance for each length $n \le 7$.
The classification is shown in Table \ref{tbl:inequiv SO codes over I2} below.
Codes marked with $*$ are twisted cyclic codes over $I_2$.

\begin{table}[H]
\centering
\setlength{\tabcolsep}{12.0pt}
\renewcommand{\arraystretch}{1.05}
{\fontsize{10pt}{11.2pt}\selectfont
\caption{Number of Permutation-Inequivalent Nonzero Cyclic Codes over $I_2$ of Length $n \le 7$}  \label{tbl: nonzero cyclic codes over I2 summary}
\scalebox{1.00}{
\begin{tabular}{ccccc}
\toprule
$n$
& $\begin{array}{c}
\text{\#Cyclic}  \\  \text{Codes}
\end{array}$
& \#Untwisted
&
\#Twisted
& $d_{\max}$\footnotemark[2]   \\
\toprule
1  & 2 & 2 & 0 & 1  \\
\midrule
2  & 6 & 5 & 1 & 2  \\
\midrule
3  & 8 & 8 & 0 & 3  \\
\midrule
4  & 20 & 14 & 6 & 4  \\
\midrule
5  & 8 & 8 & 0 & 5  \\
\midrule
6  & 45 & 35 & 10 & 6  \\
\midrule
7  & 17 & 17 & 0 & 7  \\
\bottomrule
\multicolumn{3}{l}{{\footnotesize ${}^\dagger$Highest possible minimum distance \vphantom{$\dfrac{1}{2}$}}}
\end{tabular}
}
}
\end{table}

\begingroup
\setlength{\LTleft}{\fill}
\setlength{\LTright}{\fill}
\setlength{\tabcolsep}{1.2pt}
\renewcommand{\arraystretch}{1.05}
\newcolumntype{Y}{>{\centering\arraybackslash}p{0.28\textwidth}}
\newcolumntype{Z}{>{\raggedright\arraybackslash}p{0.30\textwidth}}

\arrayrulecolor{black}   

{\fontsize{7pt}{8.2pt}\selectfont
\begin{longtable}{@{}cYcccccccZc@{}}
\caption{Permutation-inequivalent nonzero cyclic codes of length $n \le 7$ over $I_2$}\label{tbl:inequiv SO codes over I2}\\
\toprule
$\boldsymbol{n}$
& \textbf{Code $\boldsymbol{\mathcal{C}}$}
& $\boldsymbol{d(\mathcal{C})}$
& \textbf{ SO }
& \textbf{ QSD }
& \textbf{ SD }
& \textbf{ Even }
& \textbf{ Type IV }
& \textbf{ QTIV }
& \begin{tabular}[c]{@{}c@{}}\textbf{ Weight }\\\textbf{ Distribution }\end{tabular}
& \textbf{ $\boldsymbol{\Phi ( \mathcal{C})}$ }\\
\midrule
\endfirsthead

\multicolumn{10}{@{}l}{\tablename~\thetable{} continued from previous page}\\
\toprule
$\boldsymbol{n}$
& \textbf{Code $\boldsymbol{\mathcal{C}}$}
& $\boldsymbol{d(\mathcal{C})}$
& \textbf{SO}
& \textbf{QSD}
& \textbf{SD}
& \textbf{Even}
& \textbf{Type IV}
& \textbf{QTIV}
& \begin{tabular}[c]{@{}c@{}}\textbf{Weight}\\\textbf{Distribution}\end{tabular}
& \textbf{ $\boldsymbol{\Phi ( \mathcal{C})}$ }\\
\midrule
\endhead

\midrule
\multicolumn{10}{r@{}}{Continued on next page}\\
\endfoot

\bottomrule
\endlastfoot
$1$
& $\begin{pmatrix*}[c]    \mathtt{b}     \end{pmatrix*}$
& $1$
& $\checkmark$
& $\checkmark$
&
&
&
&
& $[\langle0,1\rangle,\allowbreak\langle1,1\rangle]$
& $[2,1,1]$    \\

\cmidrule(lr){2-11}

& $\begin{pmatrix*}[c] \mathtt{a}  \\   \mathtt{b}    \end{pmatrix*}$
& $1$
&
&
&
&
&
&
& $[\langle0,1\rangle,\allowbreak\langle1,3\rangle]$
& $[2,2,1]$    \\

\midrule[\heavyrulewidth]

$2$
& $\begin{pmatrix*}[c]    \mathtt{b}   &  \mathtt{b}    \end{pmatrix*}$
& $2$
& $\checkmark$
&
&
& $\checkmark$
&
&
& $[\langle0,1\rangle,\allowbreak\langle2,1\rangle]$
& $[4,1,2]$
\\

\cmidrule(lr){2-11}

& $\begin{pmatrix*}[c]    \mathtt{b}   &  0  \\   0  &   \mathtt{b}    \end{pmatrix*}$
& $1$
& $\checkmark$
& $\checkmark$
&
&
&
&
& $[\langle0,1\rangle,\allowbreak\langle1,2\rangle,\allowbreak\langle2,1\rangle]$
& $[4,2,1]$
\\

\cmidrule(lr){2-11}

& $\begin{pmatrix*}[c]    \mathtt{a}   &  \mathtt{a}  \\   \mathtt{b}  &   \mathtt{b}    \end{pmatrix*}$
& $2$
& $\checkmark$
& $\checkmark$
&
& $\checkmark$
& $\checkmark$
& $\checkmark$
& $[\langle0,1\rangle,\allowbreak\langle2,3\rangle]$
& $[4,2,2]$
\\

\cmidrule(lr){2-11}

& $\begin{pmatrix*}[c]    \mathtt{c}   &  \mathtt{a}  \\   \mathtt{b}  &   \mathtt{b}    \end{pmatrix*}^*$
& $2$
& $\checkmark$
& $\checkmark$
&
& $\checkmark$
& $\checkmark$
& $\checkmark$
& $[\langle0,1\rangle,\allowbreak\langle2,3\rangle]$
& $[4,2,2]$
\\

\cmidrule(lr){2-11}

& $\begin{pmatrix*}[c]
\mathtt{a}   &  \mathtt{a}  \\
\mathtt{b}  &   0    \\
0  &  \mathtt{b}
\end{pmatrix*}$
& $1$
& $\checkmark$
&
& $\checkmark$
&
&
&
& $[\langle0,1\rangle,\allowbreak\langle1,2\rangle,\allowbreak\langle2,5\rangle]$
& $[4,3,1]$
\\

\cmidrule(lr){2-11}

& $\begin{pmatrix*}[c]
\mathtt{a}   &  0  \\
0  &  \mathtt{a}   \\
\mathtt{b}  &   0    \\
0  &  \mathtt{b}
\end{pmatrix*}$
& $1$
&
&
&
&
&
&
& $[\langle0,1\rangle,\allowbreak\langle1,6\rangle,\allowbreak\langle2,9\rangle]$
& $[4,4,1]$
\\
\end{longtable}
}
\endgroup

Since the Gray image of a cyclic code over $I_2$ is equivalent to a binary
quasi-cyclic code of index $2$, it is natural to compare the parameters in
Table~\ref{tbl:inequiv SO codes over I2} with known binary quasi-cyclic
codes of index $2$. The theory of quasi-cyclic codes of index $2$ over
finite fields has been studied in \cite{AbdukhalikovDzhumadildaevLing_QCcodesIndex2},
where structural classifications, duality properties, and lower bounds on
minimum distances are given. Binary double circulant codes, which form an
important subclass of index-$2$ quasi-cyclic codes, have also been studied
extensively.

For the Gray images obtained in Table~\ref{tbl:inequiv SO codes over I2},
the lengths are at most $14$. Most of the resulting binary codes have
minimum distance $1$ or $2$, and hence their optimality, when it occurs, is
not particularly significant from the viewpoint of error correction. The
examples with larger minimum distance, such as $[12,4,4]$ and $[14,6,4]$,
have parameters comparable to known binary quasi-cyclic codes of index $2$,
but they do not improve the best-known minimum distances. Therefore, the
main contribution of these Gray images is not the construction of new
record-breaking binary codes, but rather the realization of a family of
binary quasi-cyclic codes of index $2$ from cyclic codes over the
commutative non-unitary ring $I_2$. This construction gives such
quasi-cyclic codes an additional residue--torsion--twist interpretation.

\section{Conclusion}
\label{sec: conclusion}

In this paper, we investigated cyclic codes over the commutative non-unitary ring $I_2$ of order $4$. Although cyclic codes can be viewed as ideals of the quotient ring $R[x] / (x^n - 1)$ where $R$ is a finite field or a unitary ring, this is not the case for cyclic codes over the ring $I_2$ due to the absence of a multiplicative identity. By studying the associated residue and torsion codes over $\mathbb{F}_2$, we obtained structural descriptions of cyclic $I_2$-codes and established several characterizations of cyclicity in terms of the cyclic shift and the twist map.

We introduced the notions of twisted and untwisted cyclic codes over $I_2$ and showed that the cyclicity of an $I_2$-code can be characterized through a compatibility condition involving the twist map and the cyclic shift. The relationships among cyclic codes, residue codes, and torsion codes were examined in detail, and several results concerning duality were established. In particular, we proved that the dual of a cyclic code over $I_2$ is again cyclic.

We also studied Gray images of cyclic codes over $I_2$. Using suitable Gray-type maps, we showed that the images of cyclic $I_2$-codes give rise to binary quasi-cyclic codes of index $2$ as well as additive cyclic codes over $\mathbb{F}_4$. These connections provide a useful bridge between codes over non-unitary rings and more classical coding-theoretic structures over finite fields.

Finally, computational classifications of permutation inequivalent cyclic codes over $I_2$ for small lengths were obtained using \texttt{MAGMA}. The computations illustrate the existence of both twisted and untwisted cyclic codes and provide concrete examples supporting the theoretical results developed in this paper.

Possible directions for future work include the study of cyclic and constacyclic codes over more general non-unitary rings, the investigation of self-dual and self-orthogonal cyclic codes over $I_p$, and the construction of quantum error-correcting codes arising from Gray images of cyclic codes over non-unitary rings.

\section*{Declaration}
\textbf{Conflict of Interest} The authors declare that they have no conflict of interest in the manuscript.

\end{document}